\documentclass[12pt]{article}
\usepackage{amssymb,amsmath,amsfonts,amsthm}
\usepackage{graphicx}
\usepackage{lscape}

\usepackage[round]{natbib}
\usepackage{float}

\usepackage{enumitem}

\usepackage[usenames,dvipsnames]{color}
\usepackage[colorlinks = true,urlcolor = blue,citecolor = blue,breaklinks]{hyperref}

\newcommand{\blue}{\color[rgb]{0,0,1}}

\numberwithin{equation}{section}
\numberwithin{figure}{section}
\numberwithin{table}{section}

\newtheorem{theorem}{Theorem}[section]

\usepackage{setspace}
\usepackage[bottom]{footmisc}

\begin{document}

\vspace*{-5mm}

\begin{center}
{\Large\bf Measuring Discrete Risks on Infinite Domains: Theoretical Foundations,
Conditional Five Number Summaries, and Data Analyses}

\vspace{8mm}

{\large\sc
Daoping Yu\footnote[1]{
~{\sc Corresponding Author}: Daoping Yu, Ph.D., ASA,
Department of Mathematical Sciences, University of Wisconsin-Milwaukee,
Milwaukee, Wisconsin, USA. ~~ {\em e-mail\/}: ~{\blue\tt dyu@uwm.edu}}}

\vspace{1mm}

{\em University of Wisconsin-Milwaukee}

\vspace{5mm}

{\large\sc
Vytaras Brazauskas\footnote[2]{
~Vytaras Brazauskas, Ph.D., ASA, Department of Mathematical Sciences,
University of Wisconsin-Milwaukee, Milwaukee, Wisconsin, USA.
~~ {\em e-mail\/}: ~{\blue\tt vytaras@uwm.edu}}}

\vspace{1mm}

{\em University of Wisconsin-Milwaukee}

\vspace{5mm}

{\large\sc
Ri{\v{c}}ardas Zitikis\footnote[3]{
~Ri{\v{c}}ardas Zitikis, Ph.D., School of Mathematical and Statistical
Sciences, Western University, London, Ontario, Canada.
~~ {\em e-mail\/}: ~{\blue\tt zitikis@stats.uwo.ca}}}

\vspace{1mm}

{\em Western University}

\end{center}

\vspace{1mm}

\begin{quote}
{\small
{\bf\em Abstract\/}.
To accommodate numerous practical scenarios, in this paper we extend statistical inference for smoothed quantile estimators from finite
domains to infinite domains. We accomplish the task with the help of a newly designed truncation methodology for discrete loss distributions
with infinite domains. A simulation study illustrates the methodology in the case of  several distributions, such as Poisson, negative binomial,
and their zero inflated versions, which are commonly used in insurance
industry to model claim frequencies. Additionally, we propose a very
flexible bootstrap-based approach for the use in practice. Using
automobile accident data and their modifications, we compute what we have termed the
conditional five number summary (C5NS) for the tail risk and construct
confidence intervals for each of the five quantiles making up C5NS, and then calculate the tail probabilities. The results show that the smoothed quantile approach classifies the tail riskiness of portfolios not only more accurately  but also produces lower coefficients of variation in the estimation of tail probabilities than those obtained using the linear interpolation
approach.

\medskip

{\bf\em Keywords\/}. Bootstrap; Claim Counts; Smoothed Quantiles; Value-at-Risk; Truncated Distributions.
}
\end{quote}

%
%
%

\section{Introduction and Motivation}
\label{Section 1}

The Value-at-Risk (VaR) has been a prominent risk measure in
the insurance and banking sectors \citep[e.g.,][]{B2019}. In the case of
insurance losses, which are non-negative random variables, the VaR at
level $p\in (0,1)$, which is close to $1$ and is often set by regulators
\citep[e.g.,][]{B2019}, is the smallest capital needed to cover the losses
with probabilities not smaller than $p$:
\[
\textrm{VaR}_p(Y)
=F_Y^{-1}(p)
=\inf\big\{ y : F_Y(y)\ge p \big\},
\]
where $Y$ is the loss random variable and $F_Y$ is its cumulative
distribution function (c.d.f.). An analogous formula and its interpretation
hold in the case of real-valued profit-and-loss (P\&L) variables, with
the losses now being on the ``negative'' side of the real line.


Obviously, the VaR does not tell us what happens far in the tail. Moreover, if
coherency is  important, then the VaR lacks this property. The Expected
Shortfall (ES), on the other hand, which is also known as the conditional tail expectation as well as  by several other names, is a coherent risk measure \citep{ADEH99}, whose first-of-the-kind axiomatic characterization has been provided by \cite{WZ2021}. The ES gives the researcher a much-needed hint of what is possibly  happening in the tail beyond the VaR at the pre-specified probability level $p$.

Although the aforementioned properties of the ES are attractive, one may argue that, 
due to the usual skewness of loss distributions, the (conditional)
expectation is not the best way to ``summarize'' the tail, and so one
would naturally think of using the (conditional) median, as the statistical
literature would suggest. In this way, as a replacement to the ES at the level
$p$, we naturally arrive at the VaR at the level $(1+p)/2$. This VaR,
still being just one parameter, does not provide a fuller and satisfactory
description of the tail, a fact noted by many authors, as exemplified by
the quotation:
\begin{quote}
VaR does not account for properties of the distribution beyond the
confidence level. This implies that $\textrm{VaR}_{\alpha}(X)$ may
increase dramatically with a small increase in $\alpha$. To adequately
estimate risk in the tail, one may need to calculate several VaRs
with different confidence levels. \citep[p.~283]{SSU2008}
\end{quote}

We suggest to use a vector-valued risk measure, which we call the Conditional
Five Number Summary (C5NS), defined by
\[
\textrm{C5NS}_p(Y) = \Big(
F_Y^{-1}(u_1), F_Y^{-1}(u_2), F_Y^{-1}(u_3), F_Y^{-1}(u_4), F_Y^{-1}(u_5)
\Big),
\]
where
\begin{gather*}
u_1 = 0.90p+0.10, \quad
u_2 = 0.75p+0.25, \quad
u_3 = 0.50p+0.50, \\
u_4 = 0.25p+0.75, \quad
u_5 = 0.10p+0.90.
\end{gather*}
These five $u_i$'s give rise to the conditional
$10^{\rm th}$, $25^{\rm th}$, $50^{\rm th}$, $75^{\rm th}$, and $90^{\rm th}$
percentiles, respectively, of the distribution of $Y$ above $\textrm{VaR}_p(Y)$, giving
a fairly informative description of the distributional tail of the loss
variable $Y$.

To see why the modifier ``conditional'' is natural, take for the sake of
argument the quantile $F_Y^{-1}(u_3)$. It is easy to check that when
the c.d.f.~$F_Y$ is continuous, the quantile coincides with what we may call
the conditional tail median (CTM) defined for any loss variable
$Y$ and any probability level $p \in (0,1)$ by
\[
\textrm{CTM}_p(Y)= \textrm{median}\big( Y \mid Y > \textrm{VaR}_p(Y)\big) .
\]
In the present paper, however, we concentrate on discrete loss random
variables, and thus the distinction between $F_Y^{-1}(u_3)$ and
$\textrm{CTM}_p(Y)$ is necessary, although on the intuitive level we
may still conveniently think of the two as carrying the same meaning.

Indeed, quite often in practice \citep[e.g.,][]{dmpw07,BDG2009}, researchers encounter
discrete loss random variables. The c.d.f.'s of these
variables are stair-case functions consisting of flat segments
as well as of jumps. We can now easily see why even an infinitesimal
decrease in the level $p$ may result in a massive decrease in the
regulatory capital, and likewise, an infinitesimal increase in the
level $p$ may result in a massive increase in the regulatory capital.
This sensitivity on $u$ is unnatural and could be hugely detrimental
to either the insurer or the regulator, or to both. The issue can be
fixed by smoothing the VaR, for which a number of methods have long
been available in the statistical literature
\citep{s77,hd82,MSS2005,wh11}. The methods have by now been adopted, modified, and explored by several
insurance-focused researchers \citep{BDG2009,abg13,bg21,br23}.

In addition to having suggested the alternative risk measure C5NS to the ES for
the sake of accommodating heavily skewed distributions, in the present
paper we offer a smoothing technique that opens up a technically-convenient
path for the development of statistical inference for the VaR at any
prescribed probability level. Even more, the technique allows the researcher to
simultaneously estimate any finite number of VaR's at whatever
probability levels might have been chosen, or imposed, thus enabling
the researcher to arrive at confidence intervals, as well as at other
statistical inference results, for the proposed vector-valued
risk measure C5NS.

In this paper we follow the methodology introduced by
\citet{wh11}, which has been extended by \citet{br23} to fully resolve
the theoretical challenges that emerge when discrete risks
reside on infinite domains. Distributions of this type
(e.g., Poisson, NB) and their zero-inflated
versions are commonly used for modeling claim frequencies.
We note at the outset that in the current paper proposed smoothing
technique differs from the traditional kernel-based approach
\citep{abg13,bg21}, where one has to {\em assume\/} a bandwidth and
the forms of a kernel. With our approach, such
assumptions can be avoided as the computational formulas of
the quantile estimators directly follow from the existing
theorems for order statistics of i.i.d.~random
variables. Moreover, the smoothed quantiles can be easily
converted to, and hence used for the estimation of, tail
probabilities. As we shall see in Section~\ref{Section 5}, such ``smooth''
estimates can reduce the variability of tail estimates up to
40-60\% when compared to those based on linear approximations
\citep[][Section 13.1]{kpw12}.

The rest of this paper is organized as follows.
In Section~\ref{Section 2}, we present a three-part design of the truncation
methodology and lay theoretical foundations for it.
In Section~\ref{Section 3}, we carry out simulation studies using the regular
Poisson and NB distributions as well as their zero-inflated
versions, in order to illustrate the established theory.
In Section~\ref{Section 4}, we design an algorithm for bootstrapping
smoothed quantiles, thus yet again validating
our theory and also serving a flexible tool for approximating
more complex problems.
In Section~\ref{Section 5}, we demonstrate the practical advantages of the new
methodology in capturing the tail risk of insurance portfolios.
Finally, a summary of the paper and concluding remarks are
provided in Section~\ref{Section 6}.

\section{Smoothed Discrete Risks}
\label{Section 2}

\subsection{Finite Domains}
\label{Section 2.1}

Consider a discrete random variable $Y$ with c.d.f.~$F_Y$ and probability mass function (p.m.f.)
$p_j = \mathbf{P} \big\{ Y = y_{j:d} \big\}$, where $y_{j:d}$ is the $j^{\rm th}$
smallest {\em distinct\/} value that $Y$ can take. Denote
$F_j := F_{Y}(y_{j:d}) = \sum_{i=1}^j p_i$, with $F_0 \equiv 0$.
When $\sum_{j=1}^d p_j = 1$ and $1 < d < \infty$ (the total number of
possible distinct values), the smoothed population quantile
function for the discrete random variable $Y$ is defined as
\begin{equation}
\label{Qu1}
Q_{Y}(u) ~=~ \sum_{j=1}^d \Big(
B_{\alpha_u, \beta_u}(F_{j}) -
B_{\alpha_u, \beta_u}(F_{j-1})
\Big) y_{j:d}
~=:~
\sum_{j=1}^d w_{j(u)} \, y_{j:d},
\end{equation}
where $B_{\alpha_u,\beta_u}$ denotes the c.d.f.~of a beta
random variable with the parameters $\alpha_u = (d+1) u$ and
$\beta_u = (d+1) (1-u)$. Note that the weights satisfy $w_{j(u)} \geq 0$
and $\sum_{j=1}^d w_{j(u)} = 1$. To gain intuitive appreciation of definition~\eqref{Qu1}, we refer to \cite{br23} and references therein.

When an i.i.d.~realization of $Y$ is obtained, say
$y_1, \ldots, y_n$, then $y_{1:d} < \cdots < y_{d:d}$ represent
the distinct data points with the corresponding frequencies
$r_{1}, \ldots, r_{d}$. The sample p.m.f.~is
$\widehat{p}_{i} = r_{i}/n$ and the empirical c.d.f.~at
$y_{j:d}$ is given by $\widehat{F}_j = \widehat{F}_{Y}(y_{j:d})
= \sum_{i=1}^j \widehat{p}_i = n^{-1} \sum_{i=1}^j r_i$. Thus,
the sample estimator of the smoothed $u^{\rm th}$ quantile for discrete
data is defined by replacing $F_j$ by $\widehat{F}_j$ in
definition~\eqref{Qu1} of $Q_Y(u)$. In this way we arrive at the estimator
\begin{equation}
\widehat{Q}_Y(u) ~=~
\sum_{j=1}^d \Big(
B_{\alpha_u, \beta_u}(\widehat{F}_{j}) -
B_{\alpha_u, \beta_u}(\widehat{F}_{j-1})
\Big) y_{j:d}
~=:~
\sum_{j=1}^d \widehat{w}_{j(u)} y_{j:d} ,
\label{empQu1}
\end{equation}
where $\widehat{F}_0 = 0$ and $B_{\alpha_u,\beta_u}$ is the beta
c.d.f.~with $\alpha_u = (d+1) u$ and $\beta_u = (d+1) (1-u)$.
As proven by \citet[][Theorem 4.1]{wh11},
$\widehat{Q}_Y(u)$ consistently estimates $Q_Y(u)$ and is asymptotically
normal. Also,
$\widehat{Q}_Y(u_1), \ldots, \widehat{Q}_Y(u_l)$ are consistent
and jointly asymptotically normal \cite[Theorem 3.1]{br23}.

\subsection{Infinite Domains}
\label{Section 2.2}

The infinite case $d = \infty$ includes many relevant distributions used
in actuarial research. For example, the Poisson, NB, and their
zero-inflated versions are commonly used for modeling claim frequencies.
To replicate the design and properties of the estimators introduced in Section~\ref{Section 2.1},
\citet[][Section 3.4]{br23} proposed to construct truncated versions of
infinitely countable discrete distributions. The goal was to find a finite
interval where most of the probability mass would be located, and then emulate
the finite case $d < \infty$.

Specifically, denote the mean and variance
of $Y$ by $\mu_Y$ and $\sigma_Y^2$, respectively,
assuming $\sigma_Y^2 < \infty$. Define the
interval
\begin{equation}
\label{chebyshev}
\big[ L_k; \; U_k \big] ~:=~
\big[ \mu_Y - k \sigma_Y; \;
\mu_Y + k \sigma_Y \big].
\end{equation}
According to Chebyshev's inequality, the probability that $Y$
 falls into interval~\eqref{chebyshev} is at least $1-1/k^2$.
Moreover, when $\mu_Y$ and $\sigma_Y^2$ are
estimated with their respective sample versions $\overline{Y}$
and $S^2$ (based on a sample of size $n$), the coverage probability
bound remains fairly close to $1-1/k^2$ and is equal to
$1 - \frac{1}{n+1}
\left[
\frac{n+1}{n} \left( \frac{n-1}{k^2} + 1 \right)
\right]$,
where $[\cdot]$ denotes the greatest integer part \citep[see][]{k12}.
For example, if $k=5$, then this empirical bound is equal to 0.909
for $n=10$, 0.941 for $n=50$, and 0.950 for $n=100$, whereas
Chebyshev's bound is equal to 0.960. \citet{br23} considered multiple
choices of $k$ in their simulation studies and recommended $k = 3, 4$,
or 5 as the most reasonable practical choices.

With this in mind,
the total number of distinct points and the smallest and largest
distinct point used in definitions~\eqref{Qu1} and \eqref{empQu1} were defined
as follows:
\begin{equation}
\label{pop-int-old}
\mbox{Population:}
\qquad
d_{k} ~=~ y_{d_k:d_k} - y_{1:d_k} + 1,
\quad
y_{1:d_k} ~=~
\max \left\{ 0, \big[ L_k \big] \right\},
\quad
y_{d_k:d_k} ~=~  \big[ U_k \big] + 1,
\end{equation}
\begin{equation}
\label{emp-int-old}
\mbox{Sample:}
\qquad
\widehat{d}_{k} ~=~
\widehat{y}_{\widehat{d}_{k}:\widehat{d}_{k}} -
\widehat{y}_{1:\widehat{d}_{k}} + 1,
\quad
\widehat{y}_{1:\widehat{d}_{k}} ~=~ \max \left\{ 0,
\big[ \widehat{L}_k \big] \right\},
\quad
\widehat{y}_{\widehat{d}_{k}:\widehat{d}_{k}} ~=~
\big[ \widehat{U}_k \big] + 1,
\end{equation}
where $[\cdot]$ denotes the greatest integer part,
$L_k$ and $U_k$ are given by definition~\eqref{chebyshev}, and
$\widehat{L}_k$ and $\widehat{U}_k$ are the sample
versions of $L_k$ and $U_k$, respectively, that is, when $\mu_Y$
is replaced by $\overline{Y}$ and $\sigma_Y^2$
by $S^2$. The corresponding truncated distributions
are
\begin{equation}
\label{popTR}
F^*_{j(k)} ~=~ \mathbf{P}
\left\{
Y \leq y_{j:d_{k}}  \mid  y_{0:d_{k}} <
Y \leq y_{d_{k}:d_{k}}
\right\}
~=~
\frac{F_{j(k)} - F_{0(k)}}{F_{d_k(k)} - F_{0(k)}}
\end{equation}
and
\begin{equation}
\label{empTR}
\widehat{F}^*_{j(k)} ~=~
\widehat{\mathbf{P}}
\left\{
Y \leq \widehat{y}_{j:\widehat{d}_{k}}  \mid
\widehat{y}_{0:\widehat{d}_{k}} <
Y \leq \widehat{y}_{\widehat{d}_{k}:\widehat{d}_{k}}
\right\}
~=~
\frac{\widehat{F}_{j(k)} - \widehat{F}_{0(k)}}
{\widehat{F}_{\widehat{d}_k(k)} - \widehat{F}_{0(k)}} \, .
\end{equation}

This design of truncated distributions as specified by quantities~\eqref{pop-int-old}--\eqref{empTR} is easy to implement
in practice but difficult to work with when theoretical
properties of the quantile estimators are considered.
Indeed, it is not
clear how to prove the conjecture about the asymptotic properties of
such estimators \citep[see][Conjecture 3.1]{br23} because points~\eqref{pop-int-old} and
\eqref{emp-int-old} use the rounding down, or ``floor,''
operation and the truncated distributions~\eqref{popTR} and \eqref{empTR} involve
c.d.f.'s that are discontinuous at those points.

Thus, we propose to revise the design of the finite intervals
and the corresponding truncated distributions as follows:
\begin{enumerate}
  \item Use the same definitions of $L_k$ and $U_k$, as well as of
$\widehat{L}_k$ and $\widehat{U}_k$ as before. To make sure
they always result in non-integer values, choose $k$ an irrational number.
For example, to get $k$ greater than 3, we may consider
$k = \pi \approx 3.1416$,
$k = \pi^2 \approx 9.8696$, or
$k = \pi^3 \approx 31.0063$, which yield the following values
of Chebyshev's bound: 0.899, 0.990, and 0.999, respectively.
  \item
  Replace points~\eqref{pop-int-old} and \eqref{emp-int-old}
with
\begin{equation}
\label{pop-int}
y_{1:d_k} ~=~
\min \big\{ y_j \in \mathbb{Z} ~\big|~ y_j \in [L_k; U_k] \big\},
\qquad
y_{d_k:d_k} ~=~
\max \big\{ y_j \in \mathbb{Z} ~\big|~ y_j \in [L_k; U_k] \big\},
\end{equation}
\begin{equation}
\label{emp-int}
\widehat{y}_{1:\widehat{d}_k} ~=~
\min \big\{ \widehat{y}_j \in \mathbb{Z} ~\big|~
\widehat{y}_j \in [\widehat{L}_k; \widehat{U}_k] \big\},
\qquad
\widehat{y}_{\widehat{d}_k:\widehat{d}_k} ~=~
\max \big\{ \widehat{y}_j \in \mathbb{Z} ~\big|~
\widehat{y}_j \in [\widehat{L}_k; \widehat{U}_k] \big\}.
\end{equation}
Here $d_{k} ~=~ y_{d_k:d_k} - y_{1:d_k} + 1$ (population)
and $\widehat{d}_{k} ~=~ \widehat{y}_{\widehat{d}_k:\widehat{d}_k} -
\widehat{y}_{1:\widehat{d}_k} + 1$ (sample). Note that
$y_j$ and $\widehat{y}_j$ represent distinct integers which
may be negative and outside the support of the underlying
probability distribution. Fortunately, such situations do not create issues.
For example, if $L_k < 0$, then there is no probability mass
on $y_{1:d_k}, \ldots, y_{r:d_k}$ for some $1 \leq r < d_k$.
This would result in zero weights assigned to points
$y_{1:d_k}, \ldots, y_{r:d_k}$ in definition~\eqref{Qu1} of $Q_Y(u)$, and would reduce
the number $d_k$ to $d_k-r$. Similar explanation applies to
the case $\widehat{L}_k < 0$.
  \item
  Replace truncated distributions~\eqref{popTR} and \eqref{empTR}
with
\begin{equation}
\label{popTRnew}
F^*_{j(k)} ~=~ \mathbf{P}
\left\{
Y \leq y_{j:d_{k}}  \mid  L_k < Y \leq U_k
\right\}
~=~
\frac{F_{j(k)} - F_Y(L_k)}{F_Y(U_k) - F_Y(L_k)}
\end{equation}
and
\begin{equation}
\label{empTRnew}
\widehat{F}^*_{j(k)} ~=~
\widehat{\mathbf{P}}
\left\{
Y \leq \widehat{y}_{j:\widehat{d}_{k}}  \mid
\widehat{L}_k < Y \leq \widehat{U}_k
\right\}
~=~
\frac{\widehat{F}_{j(k)} - \widehat{F}_Y(\widehat{L}_k)}
{\widehat{F}_Y(\widehat{U}_k) - \widehat{F}_Y(\widehat{L}_k)} \, .
\end{equation}
Here $F_{j(k)} = \mathbf{P} \left\{ Y \leq y_{j:d_{k}} \right\}$,
$F_Y(L_k) = \mathbf{P} \left\{ Y \leq L_k \right\}$, and
$F_Y(U_k) = \mathbf{P} \left\{ Y \leq U_k \right\}$.
For the empirical versions of these functions, we use
$\widehat{y}_{j:\widehat{d}_{k}}$,
$\widehat{L}_k$, and $\widehat{U}_k$.
Note that $F_Y$ is continuous
and differentiable at $L_k$ and $U_k$, and $\widehat{F}_Y$ is continuous
and differentiable at $\widehat{L}_k$ and
$\widehat{U}_k$, because these points are
non-integers.
\end{enumerate}

Following the above three-part design and in particular utilizing
formulas~\eqref{pop-int} and \eqref{popTRnew}, we define the
smoothed $u^{\rm th}$ quantile
\begin{equation}
\label{popQu}
Q^{(k)}_*(u)
~=~
\sum_{j=1}^{d_k} \Big(
B_{\alpha_u, \beta_u}(F^*_{j(k)}) -
B_{\alpha_u, \beta_u}(F^*_{j(k)-1})
\Big) y_{j:d_k}
~=:~ \sum_{j=1}^{d_k} w^{(k)}_{j(u)} y_{j:d_k}
\end{equation}
for the truncated discrete population,
where $B_{\alpha_u, \beta_u}$ denotes the beta c.d.f.~with the parameters $\alpha_u = (d_k+1) u$ and $\beta_u = (d_k+1) (1-u)$.
Likewise, using formulas~\eqref{emp-int} and \eqref{empTRnew}, we arrive at
the smoothed $u^{\rm th}$ quantile
\begin{equation}
\label{empQu}
\widehat{Q}^{(k)}_*(u)
~=~
\sum_{j=1}^{\widehat{d}_k} \Big(
B_{\widehat{\alpha}_u, \widehat{\beta}_u}(\widehat{F}^*_{j(k)}) -
B_{\widehat{\alpha}_u, \widehat{\beta}_u}(\widehat{F}^*_{j(k)-1})
\Big) \widehat{y}_{j:\widehat{d}_k}
~=:~
\sum_{j=1}^{\widehat{d}_k}
\widehat{w}^{(k)}_{j(u)} \widehat{y}_{j:\widehat{d}_k}
\end{equation}
for the truncated discrete sample,
where $B_{\widehat{\alpha}_u, \widehat{\beta}_u}$
is the beta c.d.f.~with the parameters
$\widehat{\alpha}_u = (\widehat{d}_k+1) u$ and
$\widehat{\beta}_u = (\widehat{d}_k+1) (1-u)$.

Importantly, the functions
$\widehat{Q}^{(k)}_*$ and $Q^{(k)}_*$ are directly related
to their non-truncated versions, $\widehat{Q}^{(k)}_Y$ and
$Q^{(k)}_Y$, respectively.
For example, $Q^{(k)}_*$ can be interpreted as the inverse
function of a smoothed version of $F^*$, which is related to
a similarly smoothed version of $F_Y$ through equation~\eqref{popTRnew}.
Inverting truncated distribution~\eqref{popTRnew} for smoothed $F^*$ and $F_Y$ leads
to the equation
\begin{equation}
\label{Qu-new}
Q^{(k)}_*(u) ~=~ Q^{(k)}_Y
\Big(
F_Y(L_k) + u \big( F_Y(U_k) - F_Y(L_k) \big)
\Big)
\end{equation}
relating $Q^{(k)}_*$ to $Q^{(k)}_Y$,
which is the inverse of smoothed $F_Y$.
It is clear from formula~\eqref{Qu-new} that if $F_Y(L_k) \approx 0$
and $F_Y(U_k) \approx 1$, then $Q^{(k)}_*(u) ~ \approx ~
Q^{(k)}_Y(u)$.

In Figure~\ref{Figure 2.1},
\begin{figure}[h!]
\centering
\resizebox{68mm}{68mm}{\includegraphics{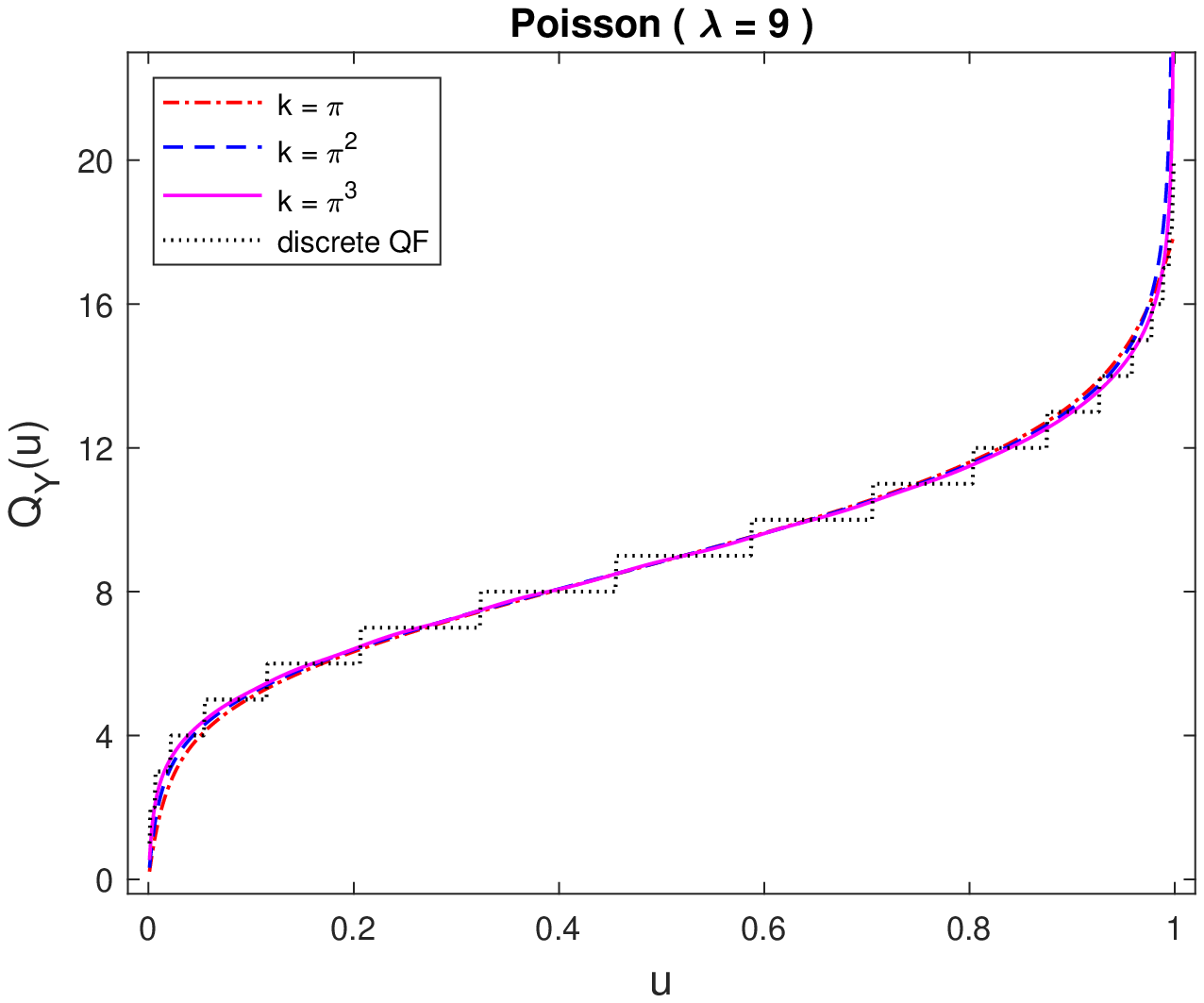}}
\qquad
\resizebox{68mm}{68mm}{\includegraphics{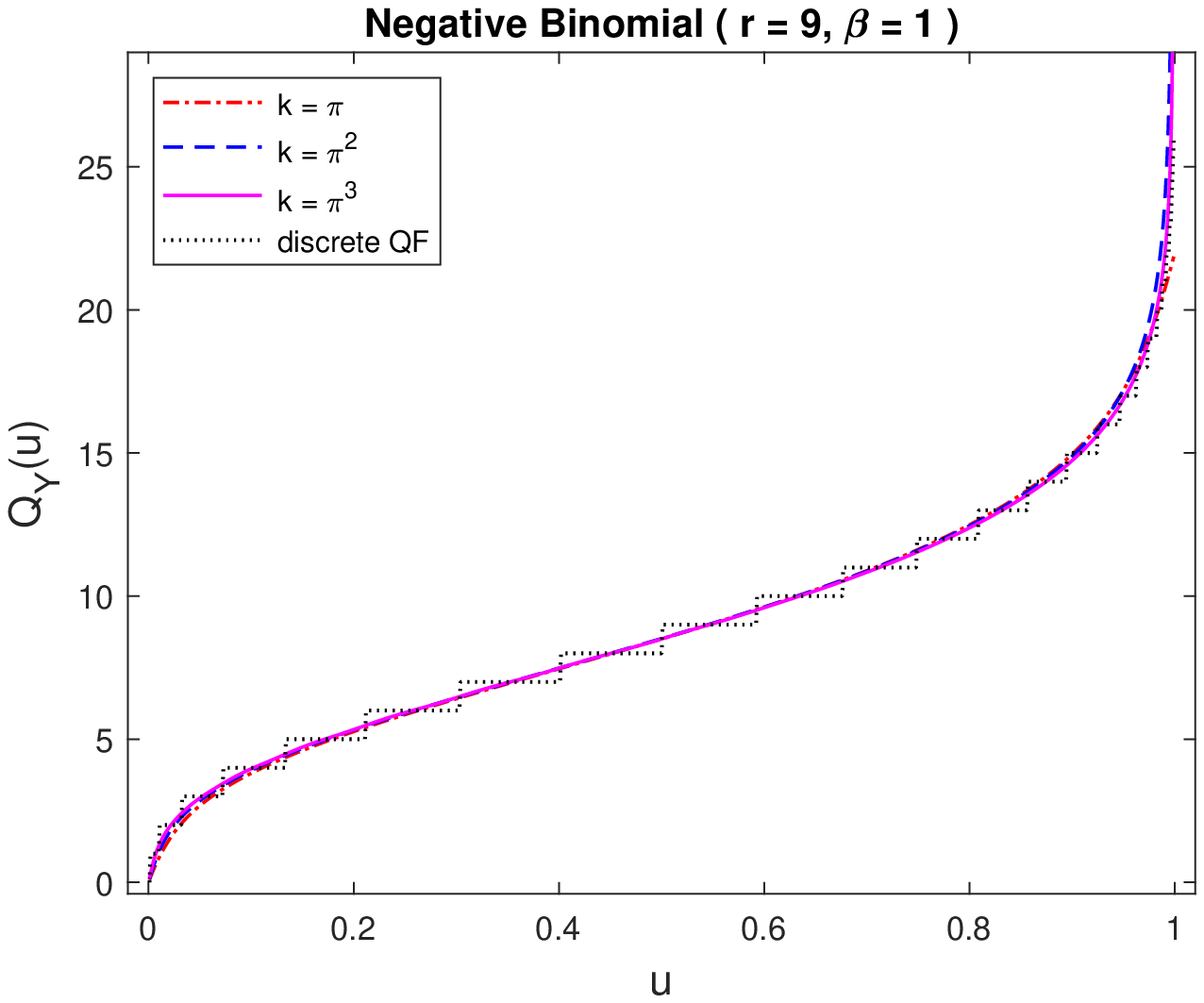}}
\\
\resizebox{68mm}{68mm}{\includegraphics{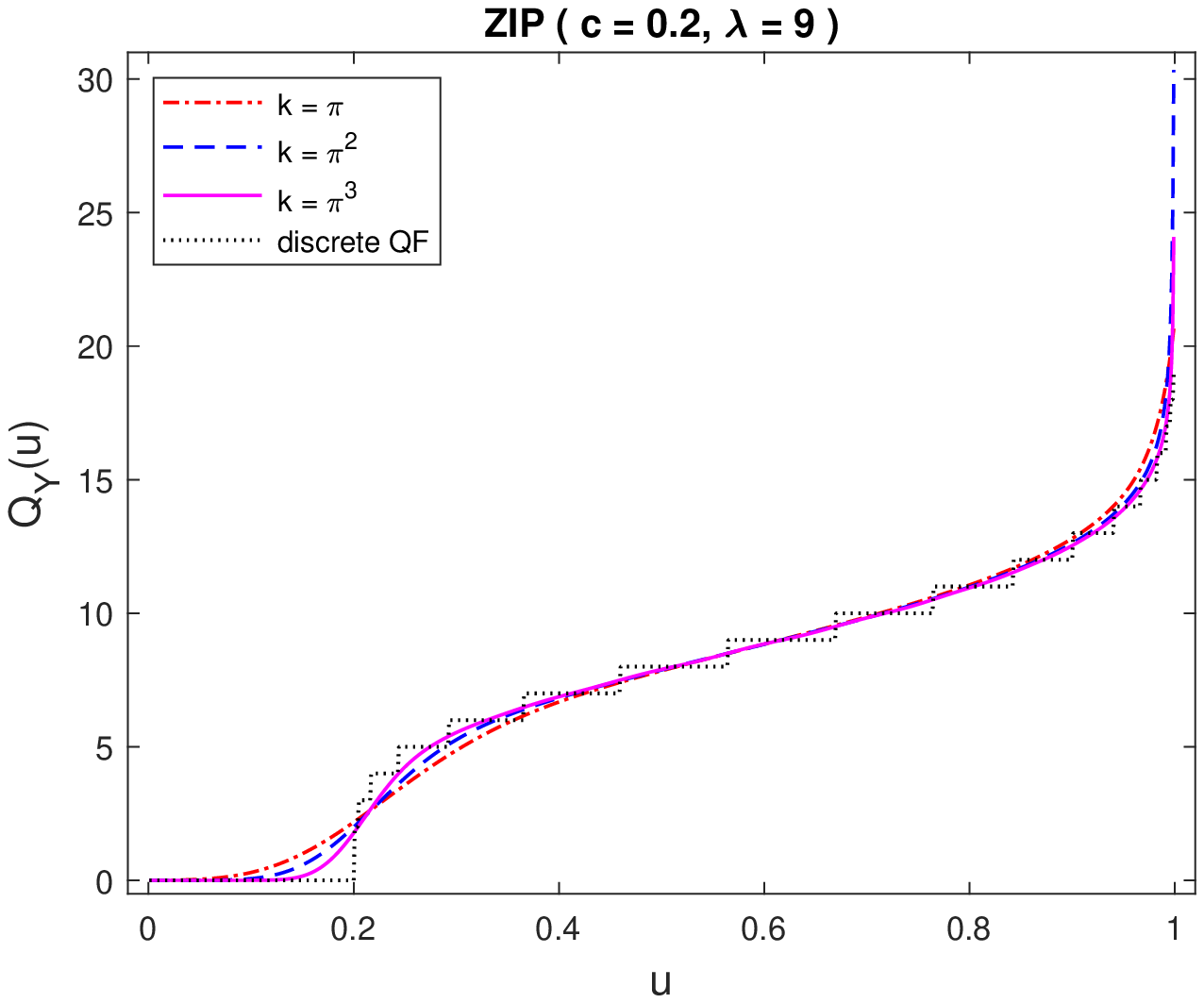}}
\qquad
\resizebox{68mm}{68mm}{\includegraphics{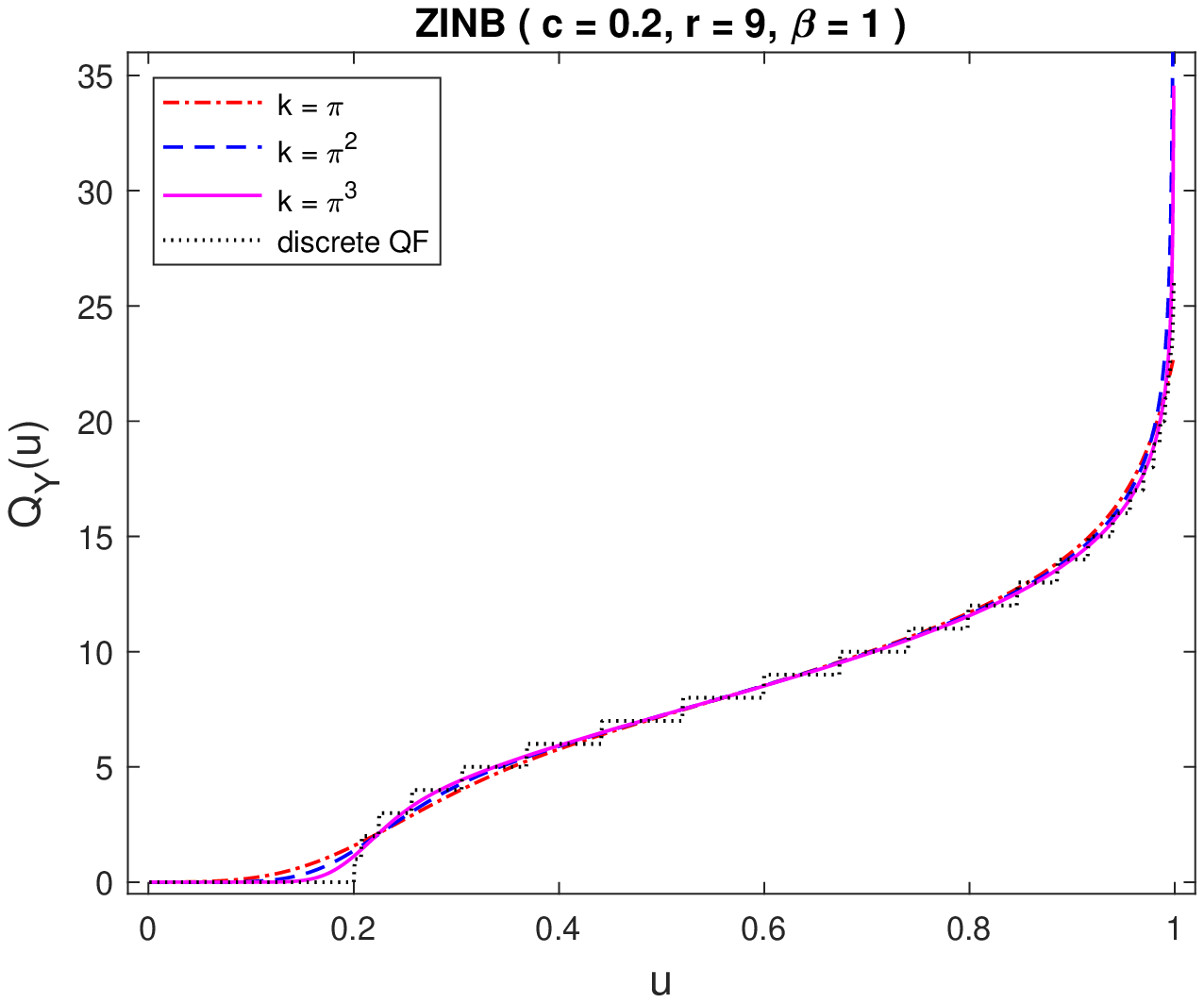}}
\\
\resizebox{68mm}{68mm}{\includegraphics{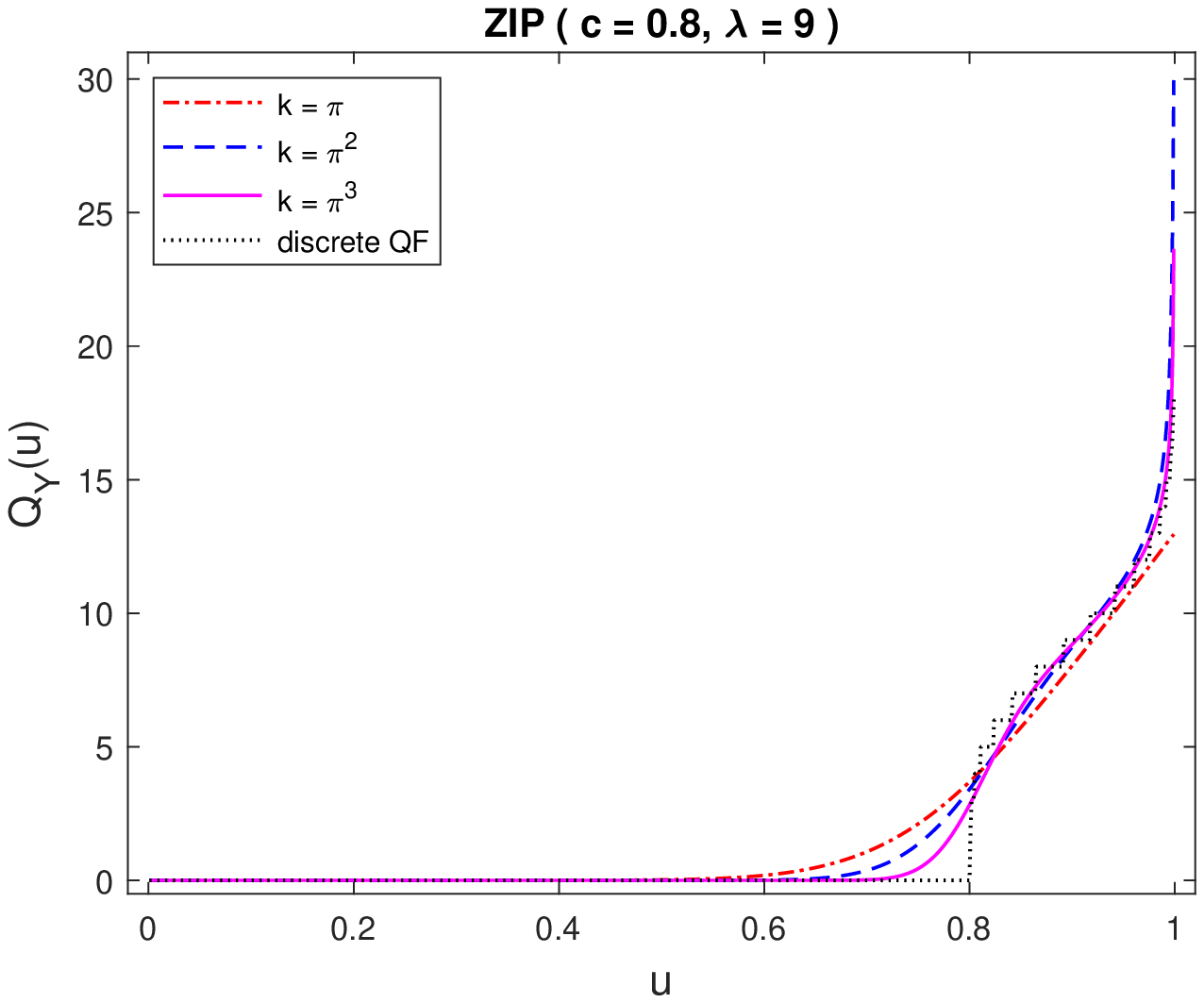}}
\qquad
\resizebox{68mm}{68mm}{\includegraphics{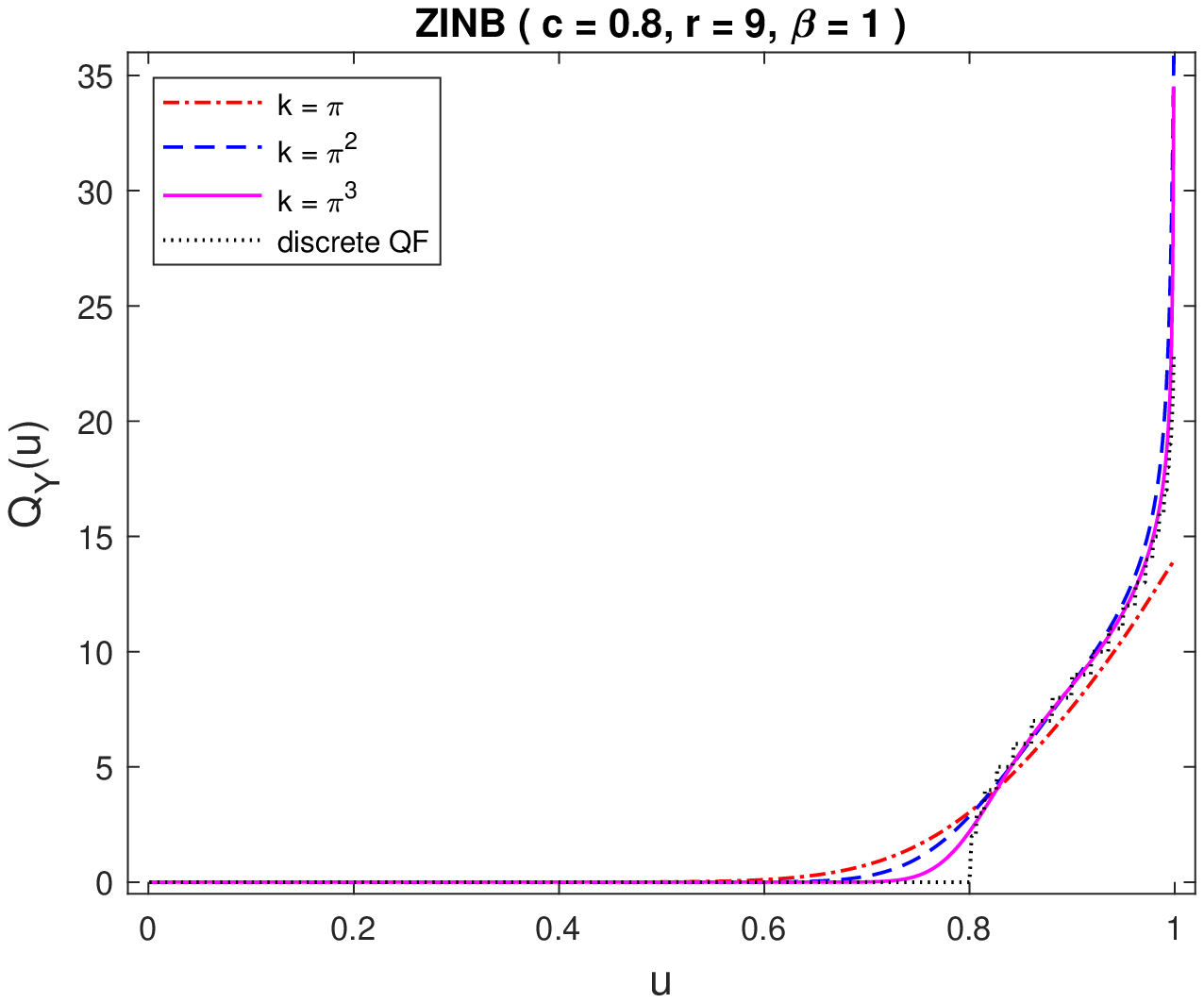}}
\caption{Discrete (dotted) and smooth quantile functions
of the Poisson (top left), NB (top right), and
the corresponding zero-inflated distributions for data
truncation intervals
$\mu_Y \pm k \sigma_Y$.}
\label{Figure 2.1}
\end{figure}
we have depicted the smooth quantile functions
$Q^{(k)}_Y(u)$ and the discrete quantile functions of
Poisson($\lambda = 9$) and NB($r = 9, \beta = 1$) distributions and their zero-inflated
versions with $c = 0.2$ and $0.8$.
The parameters $\lambda$, $r$, and $\beta$ have been selected so
that both distributions would have the same mean, but the variance
of NB would be two times larger. (For specific parametrization of these distributions, we refer to \cite{kpw12}.)
The smooth curves are constructed using data truncation intervals
$\mu_Y \pm k \sigma_Y$ with
$k = \pi, \pi^2, \pi^3$. As we see from the figure, all the
three choices of $k$ work well for standard Poisson
and NB distributions. However, for their
zero-inflated versions ZIP and ZINB the case $k = \pi$
misses both tails of the distribution and $k = \pi^2$ can
be markedly improved by $k = \pi^3$ at the jumps from 0
to 1, 2, 3, or 4. Since typical claim count data
contain about 80\% of zeros (i.e., $c = 0.8$), we recommend
using $k$ of the magnitude $\pi^3 \approx 31$.

\subsection{Theoretical Foundations}
\label{Section 2.3}

With $Q^{(k)}_*(u)$ and
$\widehat{Q}^{(k)}_*(u)$ defined by equations~\eqref{popQu}
and \eqref{empQu}, respectively, the vector
\begin{equation}
\label{popQu2}
\Big( Q^{(k)}_*(u_1), \ldots, Q^{(k)}_*(u_l) \Big)
\end{equation}
of smoothed \textit{population} quantiles ($0 < u_1 < \cdots < u_l < 1$) can be estimated by the vector
\begin{equation}
\label{empQu2}
\left( \widehat{Q}^{(k)}_*(u_1), \ldots, \widehat{Q}^{(k)}_*(u_l) \right)
\end{equation}
of smoothed \textit{sample} quantiles.
According to our next theorem, the latter vector
is a consistent and (jointly)
asymptotically normal estimator of vector~\eqref{popQu2}.

\begin{theorem}\label{Theorem 2.1}
Given an i.i.d.~sample of size $n$
from a discrete distribution $F_Y$ with infinite support $y_{1:d} <
y_{2:d} < y_{3:d} < \cdots$ and $d = \infty$, let its truncated
version $y_{1:d_k} < y_{2:d_k} < \cdots < y_{d_k:d_k}$
with $d_k < \infty$ be constructed using the three-part design of
Section~\ref{Section 2.2}. Then, when $n \rightarrow \infty$,
\begin{enumerate}[label=\rm(\roman*)]
  \item\label{part-i}
$\left( \widehat{Q}^{(k)}_*(u_1), \ldots,
\widehat{Q}^{(k)}_*(u_l) \right)
~\stackrel{\tiny \mathbf{P}}{\longrightarrow}~
\Big( Q^{(k)}_*(u_1), \ldots, Q^{(k)}_*(u_l) \Big)$,
  \item\label{part-ii}
$\displaystyle
\left( \widehat{Q}^{(k)}_*(u_1), \ldots,
\widehat{Q}^{(k)}_*(u_l) \right)
~\sim~ \mathcal{AN}
\left(
\big( Q^{(k)}_*(u_1), \ldots, Q^{(k)}_*(u_l) \big), \,
\frac{1}{n} \, \mathbf {H D H'}
\right)$,
\end{enumerate}
where
$\mathbf{D} := \big[ d_{ij} \big]_{(d_k-1) \times (d_k-1)}$
with $d_{ij} = d_{ji} = F^*_{i(k)} (1-F^*_{j(k)})$, $i(k) \leq j(k)$,
and $\mathbf{H} := \big[ h_{ij} \big]_{l \times (d_k-1)}$
with $h_{ij} = (y_{j:d_k} - y_{j+1:d_k}) \,
b_{\alpha_{u_i}, \beta_{u_i}} (F^*_{j(k)})$.
Here $b_{\alpha_{u_i}, \beta_{u_i}}$ is the beta p.d.f.~with the parameters $\alpha_{u_i} = (d_k+1) u_i$ and
$\beta_{u_i} = (d_k+1) (1-u_i)$, and the c.d.f.~$F^*_{j(k)} $ is defined by formula~\eqref{popTRnew}.
\end{theorem}

\begin{proof}
Replacing $\widehat{L}_k$ and
$\widehat{U}_k$ in equation~\eqref{emp-int} by $L_k$ and $U_k$, respectively, which are known constants,
we arrive at equation~\eqref{pop-int}. Similarly modified
equation~\eqref{empTRnew} becomes
$
\widetilde{F}^*_{j(k)} = \frac{\widehat{F}_{j(k)} -
\widehat{F}_Y(L_k)}{\widehat{F}_Y(U_k) - \widehat{F}_Y(L_k)},
$
resulting in the modification
\[
\widetilde{Q}^{(k)}_*(u) ~=~
\sum_{j=1}^{d_k} \Big(
B_{\alpha_u, \beta_u}(\widetilde{F}^*_{j(k)}) -
B_{\alpha_u, \beta_u}(\widetilde{F}^*_{j(k)-1})
\Big) y_{j:d_k}
\]
of estimator~\eqref{empQu2}.
This creates the finite domain scenario of
Section 2.1. Therefore, the vector estimator
\[
\left( \widetilde{Q}^{(k)}_*(u_1), \ldots,
\widetilde{Q}^{(k)}_*(u_l) \right)
\]
satisfies statements~\ref{part-i} and \ref{part-ii}. This can
be established by following the proof of
\citet[Theorem 3.1]{br23}.
Note first that
$\big( \overline{Y}, S^2 \big)
\stackrel{\tiny \mathbf{P}}{\longrightarrow}
\big( \mu_Y, \sigma_Y^2 \big)$
implies $\big( \widehat{L}_k, \widehat{U}_k \big)
\stackrel{\tiny \mathbf{P}}{\longrightarrow} \big( L_k, U_k \big)$,
and since by design both $L_k \in ( y_{1:d_k}-1; y_{1:d_k} )$
and $U_k \in ( y_{d_k:d_k}; y_{d_k:d_k}+1 )$ are the points
where the c.d.f.'s $\widetilde{F}^*_{j(k)}$ and $F^*_{j(k)}$ are continuous,
$\widehat{F}^*_{j(k)} \stackrel{\tiny \mathbf{P}}{\longrightarrow}
F^*_{j(k)}$.
Now note that
$\big( \widehat{Q}^{(k)}_*(u_1), \ldots,
\widehat{Q}^{(k)}_*(u_l) \big)$ is a continuous transformation
of $\widehat{F}^*_{1:d_k}, \ldots, \widehat{F}^*_{d_k:d_k}$.
Therefore, an application of the continuous mapping theorem
\citep[e.g.,][Section 1.7]{s80} implies
\[
\left( \widehat{Q}_*^{(k)}(u_1), \ldots, \widehat{Q}_*^{(k)}(u_l) \right)
\stackrel{\tiny \mathbf{P}}{\longrightarrow}
\Big( Q_*^{(k)}(u_1), \ldots, Q_*^{(k)}(u_l) \Big),
\]
which proves part~\ref{part-i} of the theorem.

To prove part~\ref{part-ii}, we note the already established result
\[
\left( \widetilde{Q}^{(k)}_*(u_1), \ldots,
\widetilde{Q}^{(k)}_*(u_l) \right)
~\sim~ \mathcal{AN}
\left(
\big( Q^{(k)}_*(u_1), \ldots, Q^{(k)}_*(u_l) \big), \,
\frac{1}{n} \, \mathbf {H D H'}
\right).
\]
Since $\widehat{Q}^{(k)}_*(u)$ is $\widetilde{Q}^{(k)}_*(u)$
with $(L_k, U_k)$ replaced by its consistent estimator
$\big( \widehat{L}_k, \widehat{U}_k \big)$, the generalized
Slutsky's theorem \citep[Section 13.1.2]{d04} assures
that estimators $\big( \widetilde{Q}^{(k)}_*(u_1), \ldots,
\widetilde{Q}^{(k)}_*(u_l) \big)$
and
$\big( \widehat{Q}^{(k)}_*(u_1), \ldots,
\widehat{Q}^{(k)}_*(u_l) \big)$
have the same asymptotic normal distribution. This concludes the proof of part~\ref{part-ii} and establishes Theorem~\ref{Theorem 2.1}.
\end{proof}

\section{Simulated Data Examples}
\label{Section 3}

In this section, we conduct a Monte Carlo simulation study with
the objective of illustrating the theoretical properties established in
Theorem~\ref{Theorem 2.1}. We start by describing the study design (Section~\ref{Section 3.1})
and then provide summarizing tables and associated with them discussions for
Poisson and NB distributions (Section~\ref{Section 3.2}) as well as
for their zero-inflated versions ZIP and ZINB (Section~\ref{Section 3.3}).

\subsection{Study Design}
\label{Section 3.1}

The study design is based on the following choices.
\begin{figure}[h!]
\begin{minipage}{17cm}
\begin{enumerate}
  \item[] \underline{\hspace{\linewidth}}

\vspace{-1ex}

  \item[] {\bf ~ Simulation Design}

\vspace{-3ex}

  \item[] \underline{\hspace{\linewidth}}

\vspace{-1ex}

\begin{itemize}
  \item {\em Discrete distributions\/}.
~Poisson$(\lambda=9)$ and negative binomial NB$(r=9, \, \beta=1)$.

  \item {\em Zero-inflated discrete distributions\/}.
~Zero-inflated Poisson, ZIP$(\lambda=1, \, c=0.8)$, and
zero-inflated negative binomial ZINB$(r=1, \, \beta=1, \, c=0.8)$.

  \item {\em Key formulas and parameters\/} (for generating data
and computing theoretical targets).
\begin{itemize}
    \item
    Means:  $\lambda = 9$ (Poisson), ~$r \beta = 9$ (NB),
~$\frac{\lambda(1-c)}{1-e^{-\lambda}} = 0.32$ (ZIP), and $\frac{r\beta(1-c)}{1-(1+\beta)^{-r}} = 0.40$ (ZINB).
    \item
    Variances: ~$\lambda = 9$ (Poisson),
~$r \beta (1+\beta) = 18$ (NB),
$\frac{1-c}{1-e^{-\lambda}} \left(
\lambda + \frac{\lambda^2 (c-e^{-\lambda})}{1-e^{-\lambda}}
\right) = 0.53$ (ZIP), and
$\frac{1-c}{1-(1+\beta)^{-r}} \left(
r\beta(1+\beta) + \frac{(r\beta)^2(c-(1+\beta)^{-r})}{1-(1+\beta)^{-r}}
\right) = 1.04$ (ZINB).
    \item Regular proportions of zeros: ~$p_0 = e^{-\lambda} = 0.0001$
(Poisson), ~$p_0 = (1+\beta)^{-r} = 0.002$ (NB),
$\frac{p_0 (1-c)}{1-p_0} = \frac{e^{-\lambda}(1-c)}{1-e^{-\lambda}} =
0.12$ (ZIP),
 and $\frac{p_0 (1-c)}{1-p_0} = \frac{(1+\beta)^{-r}(1-c)}{1-(1+\beta)^{-r}} =
0.2$ (ZINB).
    \item Excess proportions of zeros: ~$0$ (Poisson), $0$ (NB),
$\frac{c - p_0}{1-p_0} = \frac{c-e^{-\lambda}}{1-e^{-\lambda}} =
0.68$ (ZIP),
\\[0.25ex]
and $\frac{c - p_0}{1-p_0} = \frac{c-(1+\beta)^{-r}}{1-(1+\beta)^{-r}} =
0.6$ (ZINB).
\end{itemize}

\vspace{-1ex}

  \item {\em Truncation intervals\/}.
~$\overline{Y} \pm k \, S$ for simulated data and
$\mu_Y \pm k \sigma_Y$
for theoretical targets with $k = \pi, \, \pi^2, \, \pi^3$.

  \item {\em Sample sizes\/}. ~$n = 10^2, \, 10^3, \, 10^4$.
\end{itemize}

\vspace{-3ex}

  \item[] \underline{\hspace{\linewidth}}

\end{enumerate}
\end{minipage}
\end{figure}

For any given distribution, we generate 10,000 random samples
of a specified length $n$. For each sample, we estimate the
quartiles $\big( Q_*(0.25), Q_*(0.50), Q_*(0.75) \big)$ of
the distribution using formula~\eqref{empQu}. We then compute the means
and covariance-variance ($\times n$) matrices for the 10,000
estimates of the quartiles. Note that generation of the
zero-inflated data requires a two-step procedure. In the
first step, the {\em excess\/} portion of zeros is generated
from either zero or non-zero values, with the probability of
zero being ~$\frac{c-p_0}{1-p_0}$.
In the second step, the corresponding {\em regular\/} Poisson
or NB distribution is used to generate the {\em remaining\/}
portion of zeros (as well all other positive integer values),
the probability of which is now equal to ~$\frac{p_0(1-c)}{1-p_0}$.
The above two proportions of zeros sum up to $c$, the total
proportion of zeros in the sample.

\subsection{Poisson and NB Distributions}
\label{Section 3.2}

The parameters of the Poisson and NB distributions are selected
to match the models plotted in Figure~\ref{Figure 2.1}. Both distributions
have the same mean, but the variance of NB is twice the variance
of Poisson. We notice from Table~\ref{Table 3.1}
\begin{table}[h!]
\centering
\caption{Estimated means and covariance-variance
($\times n$) matrices of the smoothed quartile
estimators
$\widehat{Q}^{(k)}_*(0.25)$, $\widehat{Q}^{(k)}_*(0.50)$,
$\widehat{Q}^{(k)}_*(0.75)$ for the Poisson$(\lambda=9)$
and NB$(r=9, \, \beta=1)$
distributions truncated at the points $\overline{Y}
\pm k \, S$ with $k = \pi, \, \pi^2, \, \pi^3$.}
\medskip
{\footnotesize
\begin{tabular}{ |c|ccc|c| }
\hline
 & $n = 10^2$ & $n = 10^3$ & $n = 10^4$ & $n=\infty$ \\
\hline
\multicolumn{5}{l}{} \\[-2.5ex]
\multicolumn{5}{l}{Poisson$\, (\lambda=9$), $k = \pi$} \\
\hline
$\widehat{means}$
& (6.82, 8.84, 11.02) 
& (6.82, 8.84, 11.02) 
& (6.82, 8.84, 11.02) 
& (6.815, 8.835, 11.021) \\ 
$\widehat{\textbf{HDH'}}$
& $\begin{bmatrix} 
11.24 &  8.51 &   5.91\\
8.51 &   11.63 &   10.10\\
5.91 &   10.10 &   16.00
\end{bmatrix}$
& $\begin{bmatrix} 
11.64 &  8.65 &   5.78\\
8.65 &   11.79 &   10.20\\
5.78 &   10.20 &   16.29
\end{bmatrix}$
& $\begin{bmatrix} 
11.42 &  8.39 &   5.54\\
8.39 &   11.45 &   9.67\\
5.54 &   9.67 &   15.45
\end{bmatrix}$
& $\begin{bmatrix} 
11.367 &  8.360 &   5.539\\
8.360 &   11.497 &   9.753\\
5.539 &   9.753 &   15.478
\end{bmatrix}$ \\
\hline
\multicolumn{5}{l}{} \\[-2.5ex]
\multicolumn{5}{l}{Poisson$\, (\lambda=9)$, $k = \pi^2$} \\
\hline
$\widehat{means}$
& (6.86, 8.84, 10.98) 
& (6.86, 8.84, 10.98) 
& (6.86, 8.84, 10.98) 
& (6.856, 8.838, 10.982) \\ 
$\widehat{\textbf{HDH'}}$
& $\begin{bmatrix} 
12.05 &  8.38 &   5.71\\
8.38 &   12.28 &   9.73\\
5.71 &   9.73 &   16.52
\end{bmatrix}$
& $\begin{bmatrix} 
12.42 &  8.58 &   5.69\\
8.58 &   12.45 &   9.95\\
5.69 &   9.95 &   17.00
\end{bmatrix}$
& $\begin{bmatrix} 
12.20 &  8.34 &   5.57\\
8.34 &   12.23 &   9.62\\
5.57 &   9.62 &   16.51
\end{bmatrix}$
& $\begin{bmatrix} 
12.153 &  8.309 &   5.526\\
8.309 &   12.289 &   9.714\\
5.526 &   9.714 &   16.579
\end{bmatrix}$ \\
\hline
\multicolumn{5}{l}{} \\[-2.5ex]
\multicolumn{5}{l}{Poisson$\, (\lambda=9)$, $k = \pi^3$} \\
\hline
$\widehat{means}$
& (6.89, 8.84, 10.94) 
& (6.89, 8.85, 10.95) 
& (6.89, 8.85, 10.95) 
& (6.893, 8.853, 10.951) \\ 
$\widehat{\textbf{HDH'}}$
& $\begin{bmatrix} 
12.90 &  8.20 &   5.61\\
8.20 &   13.17 &   9.55\\
5.61 &   9.55 &   17.79
\end{bmatrix}$
& $\begin{bmatrix} 
11.32 &  7.55 &   5.08\\
7.55 &   12.00 &   9.01\\
5.08 &   9.01 &   16.91
\end{bmatrix}$
& $\begin{bmatrix} 
10.61 &  7.09 &   4.83\\
7.09 &   11.45 &   8.49\\
4.83 &   8.49 &   15.77
\end{bmatrix}$
& $\begin{bmatrix} 
10.533 &  7.033 &   4.695\\
7.033 &   11.401 &   8.415\\
4.695 &   8.415 &   15.631
\end{bmatrix}$ \\
\hline
\multicolumn{5}{l}{} \\[-2.5ex]
\multicolumn{5}{l}{NB$\, (r=9, \, \beta=1)$, $k = \pi$} \\
\hline
$\widehat{means}$
& (5.88, 8.51, 11.64) 
& (5.86, 8.50, 11.63) 
& (5.86, 8.50, 11.63) 
& (5.859, 8.504, 11.628) \\ 
$\widehat{\textbf{HDH'}}$
& $\begin{bmatrix} 
18.66 &  15.35 &   11.46\\
15.35 &   23.08 &   21.45\\
11.46 &   21.45 &   36.85
\end{bmatrix}$
& $\begin{bmatrix} 
18.48 &  15.12 &   11.22\\
15.12 &   23.17 &   21.48\\
11.22 &   21.48 &   37.13
\end{bmatrix}$
& $\begin{bmatrix} 
18.10 &  14.53 &   10.21\\
14.53 &   22.35 &   20.27\\
10.21 &   20.27 &   35.13
\end{bmatrix}$
& $\begin{bmatrix} 
18.038 &  14.458 &   10.384\\
14.458 &   22.085 &   20.054\\
10.384 &   20.054 &   34.815
\end{bmatrix}$ \\
\hline
\multicolumn{5}{l}{} \\[-2.5ex]
\multicolumn{5}{l}{NB$\, (r=9, \, \beta=1)$, $k = \pi^2$} \\
\hline
$\widehat{means}$
& (5.92, 8.52, 11.61) 
& (5.90, 8.52, 11.61) 
& (5.90, 8.52, 11.60) 
& (5.904, 8.515, 11.604) \\ 
$\widehat{\textbf{HDH'}}$
& $\begin{bmatrix} 
20.08 &  15.10 &   11.04\\
15.10 &   24.18 &   20.33\\
11.04 &   20.33 &   37.06
\end{bmatrix}$
& $\begin{bmatrix} 
19.84 &  14.90 &   10.84\\
14.90 &   24.36 &   20.41\\
10.84 &   20.41 &   37.45
\end{bmatrix}$
& $\begin{bmatrix} 
19.51 &  14.46 &   10.25\\
14.46 &   23.91 &   20.10\\
10.25 &   20.10 &   37.35
\end{bmatrix}$
& $\begin{bmatrix} 
19.552 &  14.467 &   10.507\\
14.467 &   23.833 &   20.212\\
10.507 &   20.212 &   37.975
\end{bmatrix}$ \\
\hline
\multicolumn{5}{l}{} \\[-2.5ex]
\multicolumn{5}{l}{NB$\, (r=9, \, \beta=1)$, $k = \pi^3$} \\
\hline
$\widehat{means}$
& (5.93, 8.51, 11.56) 
& (5.93, 8.51, 11.56) 
& (5.93, 8.50, 11.55) 
& (5.928, 8.504, 11.554) \\ 
$\widehat{\textbf{HDH'}}$
& $\begin{bmatrix} 
21.65 &  14.95 &   10.99\\
14.95 &   25.66 &   20.11\\
10.99 &   20.11 &   39.37
\end{bmatrix}$
& $\begin{bmatrix} 
19.16 &  14.41 &   10.20\\
14.41 &   27.72 &   20.75\\
10.20 &   20.75 &   40.10
\end{bmatrix}$
& $\begin{bmatrix} 
17.74 &  13.80 &   9.44\\
13.80 &   28.35 &   20.84\\
9.44 &   20.84 &   40.31
\end{bmatrix}$
& $\begin{bmatrix} 
17.673 &  13.777 &   9.675\\
13.777 &   28.408 &   20.920\\
9.675 &   20.920 &   40.813
\end{bmatrix}$ \\
\hline
\multicolumn{5}{l}{} \\[-2.5ex]
\multicolumn{5}{l}{
\scriptsize {\sc Note:} ~The entries for $n < \infty$ are
the averages and sample covariances of estimated quartiles.
} \\[-0.5ex]
\multicolumn{5}{l}{
\scriptsize Results are based on 10,000 simulated samples.
Standard errors of these entries are $\leq 0.01$.
} \\
\end{tabular}}
\label{Table 3.1}
\end{table}
a rapid convergence of the estimated
means and covariance-variance ($\times n$) matrices of smoothed
quartile estimators. Even the entries for $n=100$ are close to
their respective theoretical targets, which are reported in the
column $n = \infty$. Furthermore, as it could be anticipated from
Figure~\ref{Figure 2.1}, the choice of $k$ for these distributions is not
essential. We recommend choosing $k = \pi^2$ as
sufficient for most typical discrete distributions (no zero
inflated cases though). Finally, note that the estimated means
of the estimators look similar for Poisson and NB distributions
while the entries of the covariance-variance matrices differ by
a factor of (roughly) two. This is expected, and it is due to the
choice of parameters of the Poisson and NB distributions.

\subsection{ZIP and ZINB Distributions}
\label{Section 3.3}

The parameters of the ZIP and ZINB distributions are selected to match
the models depicted in Figure~\ref{Figure 2.1} (third column). The choices of
$\lambda = r \beta = 1$ and $c=0.8$ represent realistic insurance
data scenarios. We notice from Table~\ref{Table 3.2}
\begin{table}[h!]
\centering
\caption{Estimated means and covariance-variance
($\times n$) matrices of the smoothed quartile
estimators
$\widehat{Q}^{(k)}_*(0.25)$, $\widehat{Q}^{(k)}_*(0.50)$,
$\widehat{Q}^{(k)}_*(0.75)$ for the ZIP$(\lambda=1, \, c=0.8)$
and ZINB$(r=1, \, \beta=1, \, c=0.8)$
distributions truncated at the points $\overline{Y}
\pm k \, S$ with $k = \pi, \, \pi^2, \, \pi^3$.}
\medskip
{\footnotesize
\begin{tabular}{ |c|ccc|c| }
\hline
 & $n = 10^2$ & $n = 10^3$ & $n = 10^4$ & $n = \infty$ \\
\hline
\multicolumn{5}{l}{} \\[-2.5ex]
\multicolumn{5}{l}{ZIP$\, (\lambda=1, \, c=0.8$), $k = \pi$} \\
\hline
$\widehat{means}$
& (0.01, 0.10, 0.62) 
& (0.01, 0.10, 0.62) 
& (0.01, 0.10, 0.62) 
& (0.006, 0.095, 0.616) \\ 
$\widehat{\textbf{HDH'}}$
& $\begin{bmatrix} 
0.00 &  0.02 &   0.06\\
0.02 &   0.23 &   0.80\\
0.06 &   0.80 &   3.80
\end{bmatrix}$
& $\begin{bmatrix} 
0.00 &  0.02 &   0.06\\
0.02 &   0.22 &   0.69\\
0.06 &   0.69 &   2.50
\end{bmatrix}$
& $\begin{bmatrix} 
0.00 &  0.02 &   0.06\\
0.02 &   0.21 &   0.66\\
0.06 &   0.66 &   2.23
\end{bmatrix}$
& $\begin{bmatrix} 
0.001 &  0.015 &   0.044\\
0.015 &   0.150 &   0.461\\
0.044 &   0.461 &   1.522
\end{bmatrix}$ \\
\hline
\multicolumn{5}{l}{} \\[-2.5ex]
\multicolumn{5}{l}{ZIP$\, (\lambda=1, \, c=0.8$), $k = \pi^2$} \\
\hline
$\widehat{means}$
& (0.00, 0.03, 0.53) 
& (0.00, 0.03, 0.52) 
& (0.00, 0.03, 0.51) 
& (0.000, 0.026, 0.514) \\ 
$\widehat{\textbf{HDH'}}$
& $\begin{bmatrix} 
0.00 &  0.00 &   0.00\\
0.00 &   0.06 &   0.31\\
0.00 &   0.31 &   2.63
\end{bmatrix}$
& $\begin{bmatrix} 
0.00 &  0.00 &   0.00\\
0.00 &   0.05 &   0.31\\
0.00 &   0.31 &   2.88
\end{bmatrix}$
& $\begin{bmatrix} 
0.00 &  0.00 &   0.01\\
0.00 &   0.05 &   0.41\\
0.01 &   0.41 &   3.60
\end{bmatrix}$
& $\begin{bmatrix} 
0.000 &  0.000 &   0.004\\
0.000 &   0.041 &   0.318\\
0.004 &   0.318 &   2.709
\end{bmatrix}$ \\
\hline
\multicolumn{5}{l}{} \\[-2.5ex]
\multicolumn{5}{l}{ZIP$\, (\lambda=1, \, c=0.8$), $k = \pi^3$} \\
\hline
$\widehat{means}$
& (0.00, 0.00, 0.35) 
& (0.00, 0.00, 0.32) 
& (0.00, 0.00, 0.31) 
& (0.000, 0.001, 0.315) \\ 
$\widehat{\textbf{HDH'}}$
& $\begin{bmatrix} 
0.00 &  0.00 &   0.00\\
0.00 &   0.00 &   0.06\\
0.00 &   0.06 &   3.65
\end{bmatrix}$
& $\begin{bmatrix} 
0.00 &  0.00 &   0.00\\
0.00 &   0.00 &   0.02\\
0.00 &   0.02 &   4.07
\end{bmatrix}$
& $\begin{bmatrix} 
0.00 &  0.00 &   0.00\\
0.00 &   0.00 &   0.02\\
0.00 &   0.02 &   3.94
\end{bmatrix}$
& $\begin{bmatrix} 
0.000 &  0.000 &   0.000\\
0.000 &   0.000 &   0.021\\
0.000 &   0.021 &   3.400
\end{bmatrix}$ \\
\hline
\multicolumn{5}{l}{} \\[-2.5ex]
\multicolumn{5}{l}{ZINB$\, (r=1, \, \beta=1, \, c=0.8$), $k = \pi$} \\
\hline
$\widehat{means}$
& (0.00, 0.08, 0.64) 
& (0.00, 0.07, 0.65) 
& (0.00, 0.07, 0.64) 
& (0.003, 0.069, 0.642) \\ 
$\widehat{\textbf{HDH'}}$
& $\begin{bmatrix} 
0.00 &  0.01 &   0.04\\
0.01 &   0.19 &   0.79\\
0.04 &   0.79 &   4.87
\end{bmatrix}$
& $\begin{bmatrix} 
0.00 &  0.01 &   0.03\\
0.01 &   0.18 &   0.72\\
0.03 &   0.72 &   4.40
\end{bmatrix}$
& $\begin{bmatrix} 
0.00 &  0.01 &   0.04\\
0.01 &   0.17 &   0.75\\
0.04 &   0.75 &   3.76
\end{bmatrix}$
& $\begin{bmatrix} 
0.000 &  0.007 &   0.029\\
0.007 &   0.119 &   0.519\\
0.029 &   0.519 &   2.534
\end{bmatrix}$ \\
\hline
\multicolumn{5}{l}{} \\[-2.5ex]
\multicolumn{5}{l}{ZINB$\, (r=1, \, \beta=1, \, c=0.8$), $k = \pi^2$} \\
\hline
$\widehat{means}$
& (0.00, 0.02, 0.53) 
& (0.00, 0.01, 0.50) 
& (0.00, 0.01, 0.49) 
& (0.000, 0.012, 0.489) \\ 
$\widehat{\textbf{HDH'}}$
& $\begin{bmatrix} 
0.00 &  0.00 &   0.00\\
0.00 &   0.04 &   0.28\\
0.00 &   0.28 &   3.96
\end{bmatrix}$
& $\begin{bmatrix} 
0.00 &  0.00 &   0.00\\
0.00 &   0.02 &   0.24\\
0.00 &   0.24 &   4.25
\end{bmatrix}$
& $\begin{bmatrix} 
0.00 &  0.00 &   0.00\\
0.00 &   0.02 &   0.29\\
0.00 &   0.29 &   4.74
\end{bmatrix}$
& $\begin{bmatrix} 
0.000 &  0.000 &   0.001\\
0.000 &   0.014 &   0.223\\
0.001 &   0.223 &   3.781
\end{bmatrix}$ \\
\hline
\multicolumn{5}{l}{} \\[-2.5ex]
\multicolumn{5}{l}{ZINB$\, (r=1, \, \beta=1, \, c=0.8$), $k = \pi^3$} \\
\hline
$\widehat{means}$
& (0.00, 0.00, 0.33) 
& (0.00, 0.00, 0.28) 
& (0.00, 0.00, 0.27) 
& (0.000, 0.000, 0.270) \\ 
$\widehat{\textbf{HDH'}}$
& $\begin{bmatrix} 
0.00 &  0.00 &   0.00\\
0.00 &   0.00 &   0.03\\
0.00 &   0.03 &   4.83
\end{bmatrix}$
& $\begin{bmatrix} 
0.00 &  0.00 &   0.00\\
0.00 &   0.00 &   0.01\\
0.00 &   0.01 &   5.22
\end{bmatrix}$
& $\begin{bmatrix} 
0.00 &  0.00 &   0.00\\
0.00 &   0.00 &   0.00\\
0.00 &   0.00 &   5.11
\end{bmatrix}$
& $\begin{bmatrix} 
0.000 &  0.000 &   0.000\\
0.000 &   0.000 &   0.003\\
0.000 &   0.003 &   4.155
\end{bmatrix}$ \\
\hline
\multicolumn{5}{l}{} \\[-2.5ex]
\multicolumn{5}{l}{
\scriptsize {\sc Note:} ~The entries for $n < \infty$ are
the averages and sample covariances of estimated quartiles.
} \\[-0.5ex]
\multicolumn{5}{l}{
\scriptsize Results are based on 10,000 simulated samples.
Standard errors of these entries are $\leq 0.01$.
} \\
\end{tabular}}
\label{Table 3.2}
\end{table}
that convergence of the estimated means
and covariance-variance ($\times n$) matrices of smoothed quartile
estimators is not as fast as that of the Poisson and NB distributions.
It also depends on the width of the truncation interval. While the
choice of $k = \pi^2$ yields reasonable results, $k = \pi^3$ offers
an improvement. This observation agrees with the recommendation based
on Figure~\ref{Figure 2.1}. In addition, note that the estimated means of all
quartile estimators are shrinking toward zero. This is supposed
to happen because all three quartile levels are below $c = 0.8$.
Naturally, the mean estimates that are close to zero result in
similar values (almost 0) of the covariance-variance estimates.

\section{Bootstrap Approximation}
\label{Section 4}

Using simulations in Section~\ref{Section 3}, we have illustrated the
statements of Theorem~\ref{Theorem 2.1}. In the current section, we shall further harness the power of
computers and construct a bootstrap algorithm that will help us to
approximate the results of Section~\ref{Section 2}. Note that if properly designed,
bootstrap procedures can be used to approximate even more challenging
risk measurement tasks than smoothing of discrete quantiles.
In Section~\ref{Section 4.1}, the algorithm for bootstrap estimation is outlined.
In Section~\ref{Section 4.2}, the performance of the algorithm is validated and
cross-checked with Theorem~\ref{Theorem 2.1} for the Poisson, NB, ZIP, and ZINB
distributions.

\subsection{The Algorithm}
\label{Section 4.1}

The bootstrap algorithm requires specifications of the following inputs:
\begin{itemize}
\item
\textit{Data} is a sample generated by some distribution
\item
$n$ is the sample size
\item
$m$ is the number of bootstrapped resamples
\item
$k$ is the number of standard deviations in interval~\eqref{chebyshev}
\item
$u_1, \ldots, u_l$ are the quantile levels of vectors~\eqref{popQu2}
and \eqref{empQu2}
\end{itemize}
In the description of the algorithm, we use $\boldsymbol{y}$ to denote resampled data, whose sample mean and the sample standard deviation we denote by $\overline{\boldsymbol{y}}$ and $s_{\boldsymbol{y}}$, respectively. Furthermore, given any $y\in \boldsymbol{y}$, we use the notation $\mbox{freq}(y)$ for the number of $y$'s in the data set  $\boldsymbol{y}$.

\begin{figure}[h!]
\begin{minipage}{17cm}
\begin{enumerate}
  \item[] \underline{\hspace{\linewidth}}

\vspace{-1ex}

  \item[] {\bf ~ Algorithm for Bootstrap Approximation}

\vspace{-3ex}

  \item[] \underline{\hspace{\linewidth}}

\vspace{-1ex}

  \item[] {\bf Input} ~Data; ~$n$; ~$m$; ~$k$; ~$u_1, \ldots, u_l$
  \item[] \mbox{\bf{for}} ~$i=1, \ldots, m$~ \mbox{\bf{do}}
  \begin{enumerate}
    \item[] ResampledData $(=:\boldsymbol{y})$ = bootstrap(Data, with Replacement)
    \item[] $\widehat{L}_k$ = max($-0.5, \, \overline{\boldsymbol{y}}-ks_{\boldsymbol{y}}$),
\quad $\widehat{U}_k = \overline{\boldsymbol{y}}+ks_{\boldsymbol{y}}$
    \item[] $\widehat{y}_{1:\widehat{d}_k}$ = ceiling($\widehat{L}_k$),
\quad $\widehat{y}_{\widehat{d}_k:\widehat{d}_k}$ = floor($\widehat{U}_k$)
\hfill / see \eqref{emp-int} /
    \item[] $\widehat{d}_k = \widehat{y}_{\widehat{d}_k:\widehat{d}_k}-
\widehat{y}_{1:\widehat{d}_k}+1$,
\quad $\widehat{y}_{j:\widehat{d}_k} = \widehat{y}_{1:\widehat{d}_k}+j-1$
    \item[] \mbox{\bf{for}} ~$j=1, \ldots, \widehat{d}_k$~ \mbox{\bf{do}}
    \begin{enumerate}
      \item[] $\widehat{F}_{j(k)}$ =
$\left(\mbox{freq}(\widehat{y}_{1:\widehat{d}_k}) + \cdots +
\mbox{freq}(\widehat{y}_{j:\widehat{d}_k})\right) \Big/ n$
    \end{enumerate}
    \item[] \mbox{\bf{end for}}
    \item[] $\widehat{F}_{Y}(\widehat{L}_k) =
\left(\mbox{freq}\big(y\in \boldsymbol{y} \mbox{ such that }  y<\widehat{y}_{1:\widehat{d}_k}\big)\right) \Big/ n$
    \item[] $\widehat{F}_{Y}(\widehat{U}_k) = \widehat{F}_{Y}(\widehat{L}_k) +
\left(\mbox{freq}(\widehat{y}_{1:\widehat{d}_k}) + \cdots +
\mbox{freq}(\widehat{y}_{\widehat{d}_k:\widehat{d}_k})\right) \Big/ n$
    \item[] $\widehat{F}_{j(k)}^{\ast} = \left(
\widehat{F}_{j(k)}-\widehat{F}_{Y}(\widehat{L}_k) \right) \Big/
\left(\widehat{F}_{Y}(\widehat{U}_k)-\widehat{F}_{Y}(\widehat{L}_k)\right)$
\hfill / see \eqref{empTRnew} /
    \item[] \mbox{\bf{for}} ~$u=u_1, \ldots, u_l$~ \mbox{\bf{do}}
    \begin{enumerate}
      \item[] $\widehat{\alpha}_u = (\widehat{d}_k+1)u$,
\quad $\widehat{\beta}_u = (\widehat{d}_k+1)(1-u)$
      \item[] $\widehat{Q}^{(k)}_*(u) = \sum_{j=1}^{\widehat{d}_k}
\Big( B_{\widehat{\alpha}_u, \widehat{\beta}_u}(\widehat{F}^*_{j(k)}) -
B_{\widehat{\alpha}_u, \widehat{\beta}_u}(\widehat{F}^*_{j(k)-1}) \Big)
\widehat{y}_{j:\widehat{d}_k}$
\hfill / see \eqref{empQu} /
    \end{enumerate}
    \item[] \mbox{\bf{end for}}
    \item[] Store $\widehat{Q}^{(k)}_*(u)$ in $m\times l$ matrix
  \end{enumerate}
  \item[] \mbox{\bf{end for}}
  \item[] $\mbox{\bf{column Means}}(\bf{\widehat{Q}^{(k)}_*(u))}$ =
colMeans($m\times l$ matrix)
  \item[] $\bf{\widehat{\Sigma}(\widehat{Q}^{(k)}_*(u))}$ =
cov($m\times l$ matrix)

\vspace{-3ex}

  \item[] \underline{\hspace{\linewidth}}
\end{enumerate}
\end{minipage}
\end{figure}

\subsection{Validation of the Algorithm}
\label{Section 4.2}

According to the bootstrap algorithm of Section~\ref{Section 4.1}, and given a sample of
size $n$ from some distribution, the sample is empirically resampled
with replacement (bootstrapped) and the smoothed quartile estimates
are computed. This cycle is repeated 10,000 times. The means and
covariance-variance matrices of the 10,000 estimates are then
computed and summarized.
The results are reported in Tables~\ref{Table 4.1}--\ref{Table 4.2}.
\begin{table}[h!]
\centering
\caption{Bootstrapped means and covariance-variance
($\times n$) matrices of the smoothed quartile
estimators
$\widehat{Q}^{(k)}_*(0.25)$, $\widehat{Q}^{(k)}_*(0.50)$,
$\widehat{Q}^{(k)}_*(0.75)$ for the Poisson$(\lambda=9)$
and NB$(r=9, \, \beta=1)$
distributions truncated at the points $\overline{y}
\pm k \, s$ with $k = \pi, \, \pi^2, \, \pi^3$.}
\medskip
{\footnotesize
\begin{tabular}{ |c|ccc|c| }
\hline
 & $n=10^2$ & $n=10^3$ & $n=10^4$ & $n = \infty$ \\
\hline
\multicolumn{5}{l}{} \\[-2.5ex]
\multicolumn{5}{l}{Poisson$\, (\lambda=9)$, $k = \pi$} \\
\hline
$\widehat{means}$
& (7.29, 9.06, 10.98) 
& (6.87, 8.80, 11.05) 
& (6.81, 8.82, 11.08) 
& (6.815, 8.835, 11.021) \\ 
$\widehat{\textbf{HDH'}}$
& $\begin{bmatrix} 
9.51 &  6.90 &   3.79\\
6.90 &   10.49 &   7.26\\
3.79 &   7.26 &   8.46
\end{bmatrix}$
& $\begin{bmatrix} 
11.13 &  7.73 &   5.58\\
7.73 &   10.89 &   9.80\\
5.58 &   9.80 &   16.04
\end{bmatrix}$
& $\begin{bmatrix} 
11.95 &  8.58 &   5.84\\
8.58 &   11.92 &   10.45\\
5.84 &   10.45 &   16.46
\end{bmatrix}$
& $\begin{bmatrix} 
11.367 &  8.360 &   5.539\\
8.360 &   11.497 &   9.753\\
5.539 &   9.753 &   15.478
\end{bmatrix}$ \\
\hline
\multicolumn{5}{l}{} \\[-2.5ex]
\multicolumn{5}{l}{Poisson$\, (\lambda=9)$, $k = \pi^2$} \\
\hline
$\widehat{means}$
& (7.37, 9.03, 10.95) 
& (6.91, 8.78, 11.01) 
& (6.86, 8.81, 11.04) 
& (6.856, 8.838, 10.982) \\ 
$\widehat{\textbf{HDH'}}$
& $\begin{bmatrix} 
10.63 &  7.32 &   3.44\\
7.32 &   12.75 &   7.08\\
3.44 &   7.08 &   7.83
\end{bmatrix}$
& $\begin{bmatrix} 
13.30 &  8.28 &   5.71\\
8.28 &   12.09 &   9.52\\
5.71 &   9.52 &   15.87
\end{bmatrix}$
& $\begin{bmatrix} 
12.83 &  8.59 &   5.88\\
8.59 &   12.80 &   10.45\\
5.88 &   10.45 &   17.71
\end{bmatrix}$
& $\begin{bmatrix} 
12.153 &  8.309 &   5.526\\
8.309 &   12.289 &   9.714\\
5.526 &   9.714 &   16.579
\end{bmatrix}$ \\
\hline
\multicolumn{5}{l}{} \\[-2.5ex]
\multicolumn{5}{l}{Poisson$\, (\lambda=9)$, $k = \pi^3$} \\
\hline
$\widehat{means}$
& (7.44, 9.01, 10.95) 
& (6.93, 8.77, 10.97) 
& (6.89, 8.82, 11.00) 
& (6.893, 8.853, 10.951) \\ 
$\widehat{\textbf{HDH'}}$
& $\begin{bmatrix} 
12.76 &  8.25 &   3.15\\
8.25 &   15.14 &   6.36\\
3.15 &   6.36 &   6.71
\end{bmatrix}$
& $\begin{bmatrix} 
14.63 &  9.08 &   5.25\\
9.08 &   14.10 &   8.54\\
5.25 &   8.54 &   13.40
\end{bmatrix}$
& $\begin{bmatrix} 
11.29 &  7.72 &   5.11\\
7.72 &   13.10 &   9.55\\
5.11 &   9.55 &   17.36
\end{bmatrix}$
& $\begin{bmatrix} 
10.533 &  7.033 &   4.695\\
7.033 &   11.401 &   8.415\\
4.695 &   8.415 &   15.631
\end{bmatrix}$ \\
\hline
\multicolumn{5}{l}{} \\[-2.5ex]
\multicolumn{5}{l}{NB$\, (r=9, \, \beta=1)$, $k = \pi$} \\
\hline
$\widehat{means}$
& (5.68, 8.08, 10.42) 
& (5.84, 8.40, 11.78) 
& (5.82, 8.44, 11.58) 
& (5.859, 8.504, 11.628) \\ 
$\widehat{\textbf{HDH'}}$
& $\begin{bmatrix} 
17.82 &  11.60 &   8.80\\
11.60 &   13.45 &   12.47\\
8.80 &   12.47 &   25.04
\end{bmatrix}$
& $\begin{bmatrix} 
17.14 &  13.94 &   12.20\\
13.94 &   21.49 &   23.88\\
12.20 &   23.88 &   48.24
\end{bmatrix}$
& $\begin{bmatrix} 
17.21 &  13.97 &   10.42\\
13.97 &   21.70 &   20.33\\
10.42 &   20.33 &   35.85
\end{bmatrix}$
& $\begin{bmatrix} 
18.038 &  14.458 &   10.384\\
14.458 &   22.085 &   20.054\\
10.384 &   20.054 &   34.815
\end{bmatrix}$ \\
\hline
\multicolumn{5}{l}{} \\[-2.5ex]
\multicolumn{5}{l}{NB$\, (r=9, \, \beta=1)$, $k = \pi^2$} \\
\hline
$\widehat{means}$
& (5.74, 8.11, 10.25) 
& (5.89, 8.38, 11.72) 
& (5.87, 8.44, 11.55) 
& (5.904, 8.515, 11.604) \\ 
$\widehat{\textbf{HDH'}}$
& $\begin{bmatrix} 
18.13 &  10.75 &   8.17\\
10.75 &   13.65 &   11.62\\
8.17 &   11.62 &   24.02
\end{bmatrix}$
& $\begin{bmatrix} 
16.26 &  12.38 &   11.39\\
12.38 &   20.91 &   22.39\\
11.39 &   22.39 &   51.69
\end{bmatrix}$
& $\begin{bmatrix} 
17.84 &  13.48 &   10.22\\
13.48 &   22.77 &   19.74\\
10.22 &   19.74 &   37.65
\end{bmatrix}$
& $\begin{bmatrix} 
19.552 &  14.467 &   10.507\\
14.467 &   23.833 &   20.212\\
10.507 &   20.212 &   37.975
\end{bmatrix}$ \\
\hline
\multicolumn{5}{l}{} \\[-2.5ex]
\multicolumn{5}{l}{NB$\, (r=9, \, \beta=1)$, $k = \pi^3$} \\
\hline
$\widehat{means}$
& (5.79, 8.13, 10.13) 
& (5.93, 8.35, 11.67) 
& (5.90, 8.42, 11.51) 
& (5.928, 8.504, 11.554) \\ 
$\widehat{\textbf{HDH'}}$
& $\begin{bmatrix} 
16.54 &  8.81 &   6.47\\
8.81 &   13.06 &   9.81\\
6.47 &   9.81 &   20.80
\end{bmatrix}$
& $\begin{bmatrix} 
11.35 &  9.61 &   9.30\\
9.61 &   22.65 &   22.73\\
9.30 &   22.73 &   58.05
\end{bmatrix}$
& $\begin{bmatrix} 
15.05 &  12.05 &   8.99\\
12.05 &   26.48 &   20.26\\
8.99 &   20.26 &   40.64
\end{bmatrix}$
& $\begin{bmatrix} 
17.673 &  13.777 &   9.675\\
13.777 &   28.408 &   20.920\\
9.675 &   20.920 &   40.813
\end{bmatrix}$ \\
\hline
\multicolumn{5}{l}{} \\[-2.5ex]
\multicolumn{5}{l}{
\scriptsize {\sc Note:} ~The entries for $n < \infty$ are
the averages and sample covariances of estimated quartiles.
} \\[-0.5ex]
\multicolumn{5}{l}{
\scriptsize Results are based on 10,000 bootstrap resamples.
} \\
\end{tabular}}
\label{Table 4.1}
\end{table}

\begin{table}[h!]
\centering
\caption{Bootstrapped means and covariance-variance
($\times n$) matrices of the smoothed quartile
estimators
$\widehat{Q}^{(k)}_*(0.25)$, $\widehat{Q}^{(k)}_*(0.50)$,
$\widehat{Q}^{(k)}_*(0.75)$ for the ZIP$(\lambda=1, \, c=0.8)$
and ZINB$(r=1, \, \beta=1, \, c=0.8)$
distributions truncated at the points $\overline{y}
\pm k \, s$ with $k = \pi, \, \pi^2, \, \pi^3$.}
\medskip
{\footnotesize
\begin{tabular}{ |c|ccc|c| }
\hline
 & $n = 10^2$ & $n = 10^3$ & $n = 10^4$ & $n = \infty$ \\
\hline
\multicolumn{5}{l}{} \\[-2.5ex]
\multicolumn{5}{l}{ZIP$\, (\lambda=1, \, c=0.8$), $k = \pi$} \\
\hline
$\widehat{means}$
& (0.00, 0.08, 0.54) 
& (0.01, 0.09, 0.58) 
& (0.01, 0.10, 0.62) 
& (0.006, 0.095, 0.616) \\ 
$\widehat{\textbf{HDH'}}$
& $\begin{bmatrix} 
0.00 &  0.01 &   0.03\\
0.01 &   0.13 &   0.57\\
0.03 &   0.57 &   3.96
\end{bmatrix}$
& $\begin{bmatrix} 
0.00 &  0.01 &   0.04\\
0.01 &   0.14 &   0.43\\
0.04 &   0.43 &   1.46
\end{bmatrix}$
& $\begin{bmatrix} 
0.00 &  0.02 &   0.05\\
0.02 &   0.16 &   0.48\\
0.05 &   0.48 &   1.56
\end{bmatrix}$
& $\begin{bmatrix} 
0.001 &  0.015 &   0.044\\
0.015 &   0.150 &   0.461\\
0.044 &   0.461 &   1.522
\end{bmatrix}$ \\
\hline
\multicolumn{5}{l}{} \\[-2.5ex]
\multicolumn{5}{l}{ZIP$\, (\lambda=1, \, c=0.8$), $k = \pi^2$} \\
\hline
$\widehat{means}$
& (0.00, 0.03, 0.52) 
& (0.00, 0.03, 0.49) 
& (0.00, 0.03, 0.52) 
& (0.000, 0.026, 0.514) \\ 
$\widehat{\textbf{HDH'}}$
& $\begin{bmatrix} 
0.00 &  0.00 &   0.00\\
0.00 &   0.04 &   0.22\\
0.00 &   0.22 &   2.12
\end{bmatrix}$
& $\begin{bmatrix} 
0.00 &  0.00 &   0.00\\
0.00 &   0.03 &   0.19\\
0.00 &   0.19 &   1.64
\end{bmatrix}$
& $\begin{bmatrix} 
0.00 &  0.00 &   0.00\\
0.00 &   0.05 &   0.33\\
0.00 &   0.33 &   2.73
\end{bmatrix}$
& $\begin{bmatrix} 
0.000 &  0.000 &   0.004\\
0.000 &   0.041 &   0.318\\
0.004 &   0.318 &   2.709
\end{bmatrix}$ \\
\hline
\multicolumn{5}{l}{} \\[-2.5ex]
\multicolumn{5}{l}{ZIP$\, (\lambda=1, \, c=0.8$), $k = \pi^3$} \\
\hline
$\widehat{means}$
& (0.00, 0.00, 0.36) 
& (0.00, 0.00, 0.29) 
& (0.00, 0.00, 0.33) 
& (0.000, 0.001, 0.315) \\ 
$\widehat{\textbf{HDH'}}$
& $\begin{bmatrix} 
0.00 &  0.00 &   0.00\\
0.00 &   0.00 &   0.03\\
0.00 &   0.03 &   3.11
\end{bmatrix}$
& $\begin{bmatrix} 
0.00 &  0.00 &   0.000\\
0.00 &   0.00 &   0.02\\
0.00 &   0.02 &   2.57
\end{bmatrix}$
& $\begin{bmatrix} 
0.00 &  0.00 &   0.000\\
0.00 &   0.00 &   0.02\\
0.00 &   0.02 &   3.12
\end{bmatrix}$
& $\begin{bmatrix} 
0.000 &  0.000 &   0.000\\
0.000 &   0.000 &   0.021\\
0.000 &   0.021 &   3.400
\end{bmatrix}$ \\
\hline
\multicolumn{5}{l}{} \\[-2.5ex]
\multicolumn{5}{l}{ZINB$\, (r=1, \, \beta=1, \, c=0.8$), $k = \pi$} \\
\hline
$\widehat{means}$
& (0.00, 0.06, 0.46) 
& (0.00, 0.06, 0.58) 
& (0.00, 0.07, 0.64) 
& (0.003, 0.069, 0.642) \\ 
$\widehat{\textbf{HDH'}}$
& $\begin{bmatrix} 
0.00 &  0.01 &   0.02\\
0.01 &   0.08 &   0.32\\
0.02 &   0.32 &   1.72
\end{bmatrix}$
& $\begin{bmatrix} 
0.00 &  0.01 &   0.02\\
0.01 &   0.10 &   0.42\\
0.02 &   0.42 &   2.49
\end{bmatrix}$
& $\begin{bmatrix} 
0.00 &  0.01 &   0.03\\
0.01 &   0.12 &   0.53\\
0.03 &   0.53 &   2.62
\end{bmatrix}$
& $\begin{bmatrix} 
0.000 &  0.007 &   0.029\\
0.007 &   0.119 &   0.519\\
0.029 &   0.519 &   2.534
\end{bmatrix}$ \\
\hline
\multicolumn{5}{l}{} \\[-2.5ex]
\multicolumn{5}{l}{ZINB$\, (r=1, \, \beta=1, \, c=0.8$), $k = \pi^2$} \\
\hline
$\widehat{means}$
& (0.00, 0.01, 0.36) 
& (0.00, 0.01, 0.47) 
& (0.00, 0.01, 0.48) 
& (0.000, 0.012, 0.489) \\ 
$\widehat{\textbf{HDH'}}$
& $\begin{bmatrix} 
0.00 &  0.00 &   0.00\\
0.00 &   0.01 &   0.01\\
0.00 &   0.01 &   1.42
\end{bmatrix}$
& $\begin{bmatrix} 
0.00 &  0.00 &   0.00\\
0.00 &   0.01 &   0.15\\
0.00 &   0.15 &   2.52
\end{bmatrix}$
& $\begin{bmatrix} 
0.00 &  0.00 &   0.00\\
0.00 &   0.02 &   0.20\\
0.00 &   0.20 &   3.17
\end{bmatrix}$
& $\begin{bmatrix} 
0.000 &  0.000 &   0.001\\
0.000 &   0.014 &   0.223\\
0.001 &   0.223 &   3.781
\end{bmatrix}$ \\
\hline
\multicolumn{5}{l}{} \\[-2.5ex]
\multicolumn{5}{l}{ZINB$\, (r=1, \, \beta=1, \, c=0.8$), $k = \pi^3$} \\
\hline
$\widehat{means}$
& (0.00, 0.00, 0.17) 
& (0.00, 0.00, 0.25) 
& (0.00, 0.00, 0.26) 
& (0.000, 0.000, 0.270) \\ 
$\widehat{\textbf{HDH'}}$
& $\begin{bmatrix} 
0.00 &  0.00 &   0.00\\
0.00 &   0.00 &   0.00\\
0.00 &   0.00 &   1.40
\end{bmatrix}$
& $\begin{bmatrix} 
0.00 &  0.00 &   0.00\\
0.00 &   0.00 &   0.00\\
0.00 &   0.00 &   3.05
\end{bmatrix}$
& $\begin{bmatrix} 
0.00 &  0.00 &   0.00\\
0.00 &   0.00 &   0.00\\
0.00 &   0.00 &   3.45
\end{bmatrix}$
& $\begin{bmatrix} 
0.000 &  0.000 &   0.000\\
0.000 &   0.000 &   0.003\\
0.000 &   0.003 &   4.155
\end{bmatrix}$ \\
\hline
\multicolumn{5}{l}{} \\[-2.5ex]
\multicolumn{5}{l}{
\scriptsize {\sc Note:} ~The entries for $n < \infty$ are
the averages and sample covariances of estimated quartiles.
} \\[-0.5ex]
\multicolumn{5}{l}{
\scriptsize Results are based on 10,000 bootstrap resamples.
} \\
\end{tabular}}
\label{Table 4.2}
\end{table}

As we see from the tables, the bootstrap algorithm
approximates the theoretical values established in Theorem~\ref{Theorem 2.1}
reasonably well. For the Poisson and NB distributions, $n \geq 100$
with $k = \pi^2$ is sufficient in most cases. For the ZIP and ZINB
distributions, $n=100$ may be too small, even with $k = \pi^3$, but
for $n=1000$ and larger sample sizes, the algorithm performs well. Note that
to save space in the tables, the entries of the covariance-variance
matrices are multiplied by $n$. These numbers then may give the
misleading impression that the discrepancies between the bootstrap
and theoretical approximations are large, but it can be checked
that they are not.
For example, in Table~\ref{Table 4.1} for $n=100$, Poisson, $k=\pi^2$,
and the matrix entry $(3,3)$, we actually have
$7.83/100 = 0.0783$ (bootstrap) and
$16.579/100 = 0.1658$ (theoretical).
In Table~\ref{Table 4.2} for $n=1000$, ZINB, $k=\pi^3$,
and the matrix entry $(3,3)$, we actually have
$3.05/1000 = 0.0031$ (bootstrap) and
$4.155/1000 = 0.0042$ (theoretical).

\section{Real Data Examples}
\label{Section 5}

In this section, the newly developed methodology is applied to
automobile accident data set \citep{kpw12} and its three
modifications.
Specifically, in Section~\ref{Section 5.1}, the data sets are presented and described.
In Section~\ref{Section 5.2}, conditional five number summaries are computed
for the four data sets and supplemented with 95\% (pointwise)
confidence intervals.
In Section~\ref{Section 5.3}, point estimates of a few tail probabilities are
evaluated using the traditional discrete probabillity approximation
as well as the new smoothed approach.

\subsection{Data Sets}
\label{Section 5.1}

The automobile accident data \citep[][Table 6.2]{kpw12} represent
the risk profile of 9,461 insurance policies. Following the numerical
examples of \citet[][Section 5.3]{br23}, we also consider three tail
modifications of this data set. Specifically, we take 140 policies
(corresponding to about 1.5\% of the portfolio) that report 0
accidents and replace them with 140 policies that report at least
2 accidents; this results in three different scenarios. In Table~\ref{Table 5.1},
\begin{table}[h!]
\centering
\caption{Original and modified data sets consisting of the numbers of
accidents per policy.}
\medskip
\begin{tabular}{|c|ccccccccc|c|}
\hline
Data Set &
\multicolumn{9}{|c|}{Number of Accidents} & Total Number \\[-0.5ex]
 & 0 & 1 & 2 & 3 & 4 & 5 & 6 & 7 & $\geq 8$ & of Policies \\
\hline
\hline
O (original) &
7,840 & 1,317 & 239 & 42 & 14 & 4 & 4 & 1 & 0 & 9,461 \\
M1 (modified \#1) & {\em 7,700} & 1,317 & {\em 379} & 42 & 14 &
4 & 4 & 1 & 0 & 9,461 \\
M2 (modified \#2) & {\em 7,700} & 1,317 & {\em 279} & {\em 62} &
{\em 34} & {\em 24} & {\em 24} & {\em 21} & 0 & 9,461 \\
M3 (modified \#3) & {\em 7,700} & 1,317 & 239 & 42 & 14 & 4 & 4 &
{\em 141} & 0 & 9,461 \\
\hline
\end{tabular}
\label{Table 5.1}
\end{table}
the original and the three modified data sets are provided. The modified
counts of policies are {\em italicized\/}.

At first glance, M1, M2, M3 appear to be riskier portfolios than the
original data set O. Also noticeable is a progression from the least
risky (M1) to the most risky portfolio (M3). The goal of our subsequent
computations is to check if these preliminary observations are
supported by the new methodology.

\subsection{Conditional Five Number Summaries}
\label{Section 5.2}

To illustrate how the joint behavior of smoothed quantiles helps to
assess tail riskiness of portfolios, we use the data sets of Table~\ref{Table 5.1}
and perform C5NS computations. The results are reported in Table~\ref{Table 5.2}.
\begin{table}[h!]
\centering
\caption{C5NS beyond VaR$_{0.90}$ (with 95\% confidence
intervals in parentheses)
for the original and modified automobile data sets.}
\medskip
\begin{tabular}{|c|ccccc|}
\hline
Data &
\multicolumn{5}{|c|}{Summarizing Quantiles (above the VaR$_{0.90}$ level)} \\[-0.5ex]
Set & $q_{0.91}$ & $q_{0.925}$ & $q_{0.95}$ & $q_{0.975}$ & $q_{0.99}$ \\
\hline
\hline
O & 1.35 {\footnotesize $(1.28; \, 1.41)$} & 1.60 {\footnotesize $(1.51; \, 1.68)$} &
2.28 {\footnotesize $(2.14; \, 2.43)$} & 3.70 {\footnotesize $(3.48; \, 3.92)$} &
5.33 {\footnotesize $(5.15; \, 5.50)$} \\
M1 & 1.47 {\footnotesize $(1.40; \, 1.53)$} & 1.71 {\footnotesize $(1.63; \, 1.80)$} &
2.38 {\footnotesize $(2.24; \, 2.52)$} & 3.76 {\footnotesize $(3.54; \, 3.97)$} &
5.35 {\footnotesize $(5.17; \, 5.52)$} \\
M2 & 1.86 {\footnotesize $(1.76; \, 1.96)$} & 2.25 {\footnotesize $(2.13; \, 2.37)$} &
3.19 {\footnotesize $(3.05; \, 3.34)$} & 4.69 {\footnotesize $(4.56; \, 4.82)$} &
5.96 {\footnotesize $(5.89; \, 6.04)$} \\
M3 & 2.30 {\footnotesize $(2.16; \, 2.43)$} & 2.79 {\footnotesize $(2.64; \, 2.93)$} &
3.85 {\footnotesize $(3.69; \, 4.00)$} & 5.26 {\footnotesize $(5.15; \, 5.37)$} &
6.27 {\footnotesize $(6.22; \, 6.33)$} \\
\hline
\multicolumn{6}{l}{\footnotesize
{\sc Note:} ~Results are based on Theorem~\ref{Theorem 2.1},
with the truncation points $\overline{Y} \pm \pi^3 S$.} \\
\end{tabular}
\label{Table 5.2}
\end{table}

We see from the table that the C5NS approach supports the intuitive
conclusions about O, M1, M2, and M3 (see Section~\ref{Section 5.1}). Indeed, as
the 140 policies that have 0 accidents in the original data set O report higher
numbers of accidents (all of them have 2 accidents in M1 and 7 in M3),
the five quantiles used in C5NS gradually and simultaneously increase.
Also, the associated confidence intervals are relatively narrow and
in general do not overlap (except the intervals for O and M1). This
implies that the portfolios could be statistically separated and
classified as follows: O is the least risky, M1 is somewhat riskier
than O, M2 is significantly riskier than M1, and M3 is the most risky.
Of course, for this type of statistical inference the joint asymptotic
normality of the quantiles was not used. Such a result would be needed,
however, if one decided to combine the quantiles by, for example,
taking a weighted average of them.

\subsection{Tail Probabilities}
\label{Section 5.3}

To demonstrate the advantages of the smoothed approach over the commonly
used linear interpolation \citep[][Section 13.1]{kpw12}, we estimate
several tail probabilities and evaluate the standard error and the coefficient
of variation (CV) of those estimates. Specifically, for the smoothed
variable $Y^*$, we compute a tail probability by first inverting the
quantile function and then evaluating $\mathbf{P}\{Y^* > a^*\}$ directly
if $a^*$ is non-integer or by applying the 0.5 continuity correction
 if $a^*$ is an integer \citep[][Section 5.4]{br23}.
The variability of such estimates could be assessed by inverting
the results of Theorem 2.1, but we will rely on the bootstrap
algorithm (Section~\ref{Section 4.1}) which yields practically equivalent
results (Section~\ref{Section 4.2}) but is easier to implement. For the
discrete variable $Y$, the tail probability of exceeding an integer
threshold can be estimated directly. For non-integer thresholds,
linear interpolation of the probabilities at two adjacent integers
is used. For example,
$\mathbf{P}\{Y > 1.29\} = 0.71 \, \mathbf{P} \{Y > 1\} + 0.29 \,
\mathbf{P} \{Y > 2\}$. The variability measures of the estimates
are also evaluated by employing the bootstrap approach. The results
of these calculations are summarized in Table~\ref{Table 5.3}.
\begin{table}[h!]
\centering
\caption{Tail probabilities for the original and modified data sets.
Bootstrap estimates of the mean, standard deviation, and coefficient of
variation
based on the discrete ($Y$) and smoothed ($Y^*$) approximations.}
\medskip
\begin{tabular}{|c|c|cc|cc|cc|}
\hline
Data & Estimated &
\multicolumn{6}{|c|}{$\mathbf{P}\{Y > a\}$ versus $\mathbf{P}\{Y^* > a^*\}$} \\[-0.25ex]
\cline{3-8}
Set & Quantity &
$a = 0$ & $a^* = 0.5$ & $a = 0.21$ & $a^* = 0.21$ & $a = 1.29$ & $a^* = 1.29$ \\
\hline
\hline
O & Mean     & 0.172 & 0.208 & 0.142 & 0.301 & 0.025 & 0.095 \\
  & Std. Dev. & 0.004 & 0.004 & 0.003 & 0.006 & 0.001 & 0.003 \\
  & CV       & 0.023 & 0.021 & 0.023 & 0.021 & 0.057 & 0.031 \\
\hline
M1 & Mean     & 0.186 & 0.226 & 0.157 & 0.321 & 0.035 & 0.105 \\
   & Std. Dev. & 0.004 & 0.005 & 0.003 & 0.007 & 0.002 & 0.003 \\
   & CV       & 0.022 & 0.021 & 0.022 & 0.021 & 0.046 & 0.028 \\
\hline
M2 & Mean     & 0.186 & 0.226 & 0.157 & 0.318 & 0.038 & 0.122 \\
   & Std. Dev. & 0.004 & 0.004 & 0.003 & 0.004 & 0.002 & 0.003 \\
   & CV       & 0.022 & 0.016 & 0.022 & 0.014 & 0.046 & 0.021 \\
\hline
M3 & Mean     & 0.186 & 0.231 & 0.157 & 0.319 & 0.040 & 0.137 \\
   & Std. Dev. & 0.004 & 0.004 & 0.003 & 0.004 & 0.002 & 0.003 \\
   & CV       & 0.022 & 0.015 & 0.022 & 0.014 & 0.047 & 0.021 \\
\hline
\multicolumn{8}{l}{\footnotesize {\sc Note:} ~Results are based
on 1,000 bootstrapped resamples, with the truncation points
$\overline{Y} \pm \pi^3 S$.} \\
\end{tabular}
\label{Table 5.3}
\end{table}

In the table, the probabilities $\mathbf{P}\{Y > 0\}$ and $\mathbf{P}\{Y^* > 0.5\}$
measure the chance of at least one claim. The numbers $a =  a^* = 0.21$
and $a = a^* = 1.29$ represent the events of exceeding the mean and the
mean plus two standard deviations, respectively, of the number of accidents
in the original portfolio O. (Note that for O, the mean is 0.21
and the standard deviation is 0.54.) Two observations about the smoothed
quantile approach can be made: first, it is more conservative, as it yields
higher tail probability estimates than the standard discrete variable
methodology, and second, it is more precise, as the coefficients of
variation of the ``smoothed'' estimates are always smaller than those
of the ``discrete'' estimates.

\section{Summary and Concluding Remarks}
\label{Section 6}

In this paper, we have studied the simultaneous estimation of smoothed
VaR's at several quantile levels for discrete random variables
that have been applied to model insurance claim frequencies. We have generalized the theory
from finite domains to infinite domains and showed the consistency and
joint asymptotic normality of the smoothed quantile estimators for
the truncated discrete risks. Such theoretical properties have been
established by constructing non-integer values of lower and upper
bounds for the truncated underlying population.

In addition, Monte Carlo
simulation studies have been carried out to illustrate the established
theory, through an implementation of the procedure in the theoretical
design of the truncation methodology. Commonly used discrete distributions
with infinite domains such as the Poisson, NB, and their
(realistic) zero inflated versions have been investigated in simulation
studies. We have successfully illustrated the consistency and asymptotic
normality of the smoothed quantile estimators, through the convergence
of the estimated means and covariance-variance matrices to their
corresponding theoretical counterparts, although
the convergence has been slower for the zero inflated distributions than that
for the regular distributions.

Furthermore, a bootstrap
approximation has been designed to illustrate the theoretical results.
Using the approximation, given just an original sample data set from
each considered distribution, through resampling with replacement,
we were able to see the agreement between the bootstrap estimated mean vector
and the theoretical approximation of the mean vector, as well as the agreement between the bootstrap
estimated covariance-variance matrices and the theoretical
approximation of the covariance-variance matrices.

Finally, we have applied
the truncation methodology on infinite domains to the automobile
accident data and also considered three gradual modifications of
the tail of the data. Through the computation of the vector-valued risk measure C5NS along with
confidence intervals of the original data set and its three
modifications, we have shown that the smoothed quantile
estimators can accurately classify the portfolio riskiness by
adequately assessing the tail risks. To further demonstrate the
advantages of the smoothing methodology, we have compared the tail probabilities obtained
using the smoothed approach and the linear interpolation approach.
We have found that the smoothed approach results in a lower
coefficient of variation in the estimation of tail probabilities
than the linear interpolation approach.

\section*{Acknowledgments}

This research has been supported by the NSERC Alliance--MITACS Accelerate grant entitled ``New Order of Risk Management: Theory and Applications in the Era of Systemic Risk'' from the Natural Sciences and Engineering Research Council (NSERC) of Canada, and the national research organization Mathematics of Information Technology and Complex Systems (MITACS) of Canada.

\end{document}